\documentclass[12pt]{article}
\usepackage{amsfonts}
\usepackage{amsmath,amssymb,amsthm}
\usepackage{bbm}
\usepackage{natbib}
\usepackage{geometry}
\usepackage[onehalfspacing]{setspace}
\usepackage{tikz}
\usepackage{caption}
\usepackage{subcaption}
\usepackage{enumitem}
\usepackage{hyperref}

\newtheorem{theorem}{Theorem}

\newtheorem{corollary}{Corollary}

\newtheorem{lemma}{Lemma}

\newtheorem{proposition}{Proposition}

\newcommand{\dd}[1]{\mathrm{d}#1}

\newcommand{\eqrefb}[1]{(\ref*{#1})}

\begin{document}

\title{Sharing Credit for Joint Research}
\author{Nicholas Wu\thanks{Department of Economics, Yale University, New Haven, CT 06511, nick.wu@yale.edu \\ I am deeply indebted to Johannes H\"orner for invaluable feedback and discussions. I also thank Dirk Bergemann, Mira Frick, Marina Halac, Ryota Iijima, Jonathan Libgober, Bart Lipman, Larry Samuelson, Tan Gan, participants at the Yale microeconomic theory lunch seminar, and anonymous referees for helpful comments. All errors are my own. }}
\date{March 26, 2024}
\maketitle

\begin{abstract}

How can one efficiently share payoffs with collaborators when participating in risky research? First, I show that efficiency can be achieved by allocating payoffs asymmetrically between the researcher who makes a breakthrough (``winner'') and the others, even if agents cannot observe others' effort. When the winner's identity is non-contractible, allocating credit based on effort at time of breakthrough also suffices to achieve efficiency; so the terminal effort profile, rather than the full history of effort, is a sufficient statistic. These findings suggest that simple mechanisms using minimal information are robust and effective in addressing inefficiencies in strategic experimentation.

\bigskip

\noindent \textbf{JEL Codes}: C73, D83, O31, O32
\end{abstract}

\newpage

\section{Introduction}


Researchers rarely conduct work alone; this raises questions about the information necessary to ensure efficiency. The economic literature has commonly studied the implications of informational free-riding with multiple agents. Quite generally, the equilibrium is inefficient because agents exert suboptimally low effort on research, due to the information they gain from observing each other. 
This paper seeks to extend the understanding of strategic experimentation by investigating what information must be contractible to restore efficiency. I consider an environment in which researchers exert costly effort on developing a breakthrough, but do not know whether a breakthrough is possible.

I show that one piece of information sufficient to restore efficiency is the identity of the researcher that makes the breakthrough (termed the ``winner''). To show this, I first consider a full-information environment with heterogeneity in payoffs between the discoverer and non-discoverers. I find that absent contracting, equilibria are inefficient generically, except in a knife-edge case. This case requires payoff parameters to align in a specific way; the payoff externalities must be such that the continuation value of failing to make a discovery (``losing'') equals the flow opportunity cost of research. Intuitively, if the losers benefit too much from a discovery, strategic agents have an incentive to free-ride on the efforts of others; they inefficiently reduce research and give up on research projects too easily. On the other hand, if the losers suffer in the event of a discovery, strategic agents overexert effort on failing research endeavors because they are afraid of another agent making the discovery. The efficiency condition thus depends on the losers' payoffs and the opportunity cost of research, but it notably does \textit{not} depend on what the winner receives.

While payoff externalities induced by a sharing contract conditioning on the winner's identity can fix the inefficiencies induced by strategic experimentation, the structure of the sharing contract is notable. In particular, winner-take-all contracts and equal sharing are both inefficient; the efficient contract must guarantee something to the losers, but not too much or too little, and the payoffs in the efficient contract are asymmetric ex-post. Further, it is also significant that such a contract does not require the agents to observe each others' actions; that is, the same sharing contract that restores efficiency in the observable-action model still uniquely induces the efficient outcome even if agents cannot observe each others' actions. Therefore, observability of effort is not essential to restoring efficiency.

Having shown that the identity of the winner is sufficient for restoring efficiency (even without observing effort), it might seem that this information is also necessary for implementing an efficient outcome. It is not; contracting on the effort profile at the time of breakthrough is \textit{also} sufficient to restore efficiency. That is, if the identity of the winner is not contractible, then the effort profile at time of breakthrough also suffices. In particular, this implies that the full history of effort is redundant given the terminal effort profile, and further that the identity of the winner is sufficient but not necessary to restoring efficiency. Importantly for fairness considerations, contracting on the effort profile at time of breakthrough results in outcomes which are ex-post symmetric on the equilibrium path, unlike the asymmetry necessary for efficient behavior when contracting on the winner's identity.

Methodologically, I build off of the canonical model of strategic experimentation of \cite{krc2005} where multiple agents conduct research on a project that is initially unknown to be good or bad. Exerting effort on research comes at an opportunity cost. If the project is good, the project generates a conclusive breakthrough at some rate according to each researching agent's effort. Instead, if the project is bad, a breakthrough never arrives. Breakthrough brings about fixed instantaneous and continuation rewards, shared amongst the participating agents. Since I consider a general model, the Hamilton-Jacobi-Bellman equation characterizing the agent best-response problem does not always admit a differentiable solution. To resolve this, I consider viscosity solutions, use a guess-and-verify approach to confirm an equilibrium candidate, and exploit other features of the environment to rule out other equilibria.

To intuitively understand why sharing contracts are well-suited to the research environment, note that sharing contracts create an encouragement effect by altering the degree of strategic complementarity or substitutability. In particular, for an environment similar to \cite{krc2005} in which free-riding drives inefficiency, the strategic complementary induced by a sharing contract can manufacture an offsetting encouraging effect. However, as a byproduct, this implies that the contracts considered in this paper can only alter the degree of strategic complementarity or substitutability across all agents uniformly.

As a consequence, the insights of this paper do not necessarily hold in environments in which the nature of inefficiency is heterogeneous between agents. For example, with asymmetric returns to research effort, the first-best solution takes a more complex form where some agents stop experimenting at different beliefs than others. The type of sharing contracts considered in this paper fail because the ``winner'' or discovery bonus can only be calibrated to the agent with the highest returns to effort and cannot ensure efficient behavior of the other agents; in those environments, stronger contracting instruments are necessary to restore efficiency. However, the results of the paper do extend to allow for some heterogeneity; namely, if agents have heterogeneity in the measure of resources available to invest in research (but identical returns to effort), the main insights still hold. While the analysis in this paper focuses on the conclusive good-news model of experimentation, the techniques do not rely on any specific features of the good-news model beyond the Markov assumptions.\footnote{Preliminary calculations suggest that the insights also extend to other strategic experimentation environments, like the bad-news model of \cite{kr2015}.} This paper focuses on conclusive good-news primarily because arrivals of stochastic breakthroughs plausibly model the process of conducting research.

The paper is structured as follows. The next subsection reviews the related literature. Section \ref{sec:model} lays out the experimentation game. Section \ref{sec:cooperative} derives the efficient research outcome of the experimentation game. Section \ref{sec:noncoop} discusses the equilibria of the noncooperative game, where agents strategically make research effort decisions. Section \ref{sec:contracts} analyzes sharing contracts. Section \ref{sec:extensions} discusses an extension of the model and Section \ref{sec:conclusion} concludes the paper.

\subsection{Literature}

This paper builds on the strategic experimentation literature that originated with \cite{bh1999}, which considered a Brownian motion bandit problem and identified the free-riding and encouragement effects that are present in these games. Some of the techniques used in this paper, such as considering the individual agent best-response Hamilton-Jacobi-Bellman problem and rewriting the best-response policy in the value-belief space, originally appeared in \cite{bh1999}. \cite{krc2005} first introduced the exponential bandit framework for strategic experimentation, where payoffs on the risky arm arrive as lump-sums if and only if the state is good. This paper generalizes \cite{krc2005} by considering heterogenous payoff effects after the first breakthrough.

A number of papers have extended the original \cite{krc2005} model; however, this paper specifically focuses on payoff externalities of a different form relative to those in the literature. 
Theorem 1 of \cite{hkr2022} shows that when one of two conditions is met, the inefficiency arising in \cite{krc2005} disappears when weakening the Markov solution concept to strongly symmetric equilibria; either there must be a Brownian drift component to the information process, or the belief jump from a breakthrough at the efficient threshold belief must be lower than the individual belief cutoff. Importantly, their Theorem 1 also shows that the broader strongly symmetric equilibrium concept does \textit{not} remedy the inefficiency in a pure good-news environment; thus, in this paper, it is more striking that payoff externalities can restore efficiency within the stronger Markov solution concept. That is, \cite{hkr2022} show that the inefficiency arising in the good-news environment \textit{cannot} be removed by broadening the solution concept and dropping the Markov assumption.

Some other papers consider specific forms of externalities. \cite{al2016} considers a model with two research lines that are monopolizable, but only one line is risky and can bring bad news. Their paper is focused on welfare implications of hiding bad news. In contrast, this paper considers just one research line but with arbitrary payoffs (allowing for imperfectly monopolizable research) and focuses on ex ante contracts to share rewards. In another paper, \cite{t2021} studies a problem where the safe options are rival; that is, only one agent can take the safe option. In contrast, this paper assumes that externalities only arise after a breakthrough, rather than from agents competing on the safe arm.

The results that focus on contractible information relate to strategic experimentation papers that consider the role of the observability of breakthroughs, payoffs, and actions (\cite{rsv2007}, \cite{bh2011}, \cite{rsv2013}). \cite{bh2011} considers an equal payoff sharing environment with unobserved actions; in this paper, I show that the efficient contract that redistributes payoffs between winner and losers still implements efficiency even when the actions are unobserved as in \cite{bh2011}. There are a number of other papers that focus on correlation of the bandit state (\cite{kr2011}, \cite{rsv2013}), bad news (\cite{kr2015}) and L\'evy process bandits (\cite{hkr2022}). The insights in this paper allow for generalization to asymmetries in the amount of research resource available; this complements other papers that have considered asymmetry in the quality of research between players (\cite{dks2020}) and in the informational content available to players (\cite{d2018}).

Since this paper studies an environment where the first breakthrough obtains a different payoff than the other experimenters, it also relates to the economics literature on contests. The closest paper is \cite{hkl2017}, which considers public and hidden contests where a principal incentivizes agents to exert costly effort on research; however, their paper focuses on information disclosure and whether hidden equal-sharing or public winner-take-all contests result in a higher probability of breakthrough for the principal, not on whether the outcome is necessarily socially efficient. Instead, this paper focuses on payoff characterizations that result in social efficiency, rather than maximizing total effort, which was the principal's objective in \cite{hkl2017}. 

This paper is also related to the literature on efficient dynamic mechanism design. This literature primarily focuses on the social choice setting. \cite{bv2002} study the incentives for agents to acquire information about their own types in a static social choice setting, and \cite{bv2010} formulates a dynamic pivot mechanism. More relatedly, \cite{as2013} propose a VCG-like mechanism for social choice in a dynamic environment that is also budget-balanced. This paper differs from these other papers in that these other papers require an assumption of private values, which fails when there are explicit payoff externalities resulting from experimentation. Although the \cite{as2013} model can capture informational externalities via the evolution of the belief stochastic process, the payoff externalities in this paper fail the assumption of private values because they introduce interdependence in the instantaneous incentives.

Indeed, in the presence of payoff interdependence, \cite{jm2003} show that even in a static setting, efficiency may not attain. In such environments, \cite{m2004} shows that requiring transfers after uncertainty resolution can restore efficiency. However, those insights do not apply to the dynamic experimentation setting; because experimentation stops with positive probability, there are outcomes where the state of the world never fully realizes to the agents.
 
\section{Game Structure}
\label{sec:model}
I first formally lay out the baseline structure of the research game. 

\paragraph{Exposition} 
There are $N$ agents $i \in \{1, 2, \dots N \}$ investigating a potential research breakthrough. The research idea is good or bad, which is drawn by Nature prior to the start of the game and unobserved by the agents. Formally, the quality of the research idea is the state of the world, $\omega \in \Omega := \{ \textnormal{good}, \textnormal{bad} \}$. Nature draws the state of the world to be good with probability $p(0)$, which is the initial prior belief shared by the agents on the state of the world. Time is continuous, $t \in [0, \infty)$, and at every instant of time, each agent is endowed with a unit measure of a resource (effort) that it allocates over two projects, the status quo technology or the research process. 

\paragraph{Actions}
At each instant in time, agent $i$ chooses how much effort $k_i \in [0,1]$ to allocate to the research process, with the remaining effort allocated to an outside project, which produces a flow reward according to the status quo technology.
The status quo technology yields a constant, deterministic flow payoff $\pi_s(1 - k_i)$ to the agent, where $\pi_s \ge 0$ is the flow profit per unit effort. 
The research process yields no flow payoff but could produce a breakthrough depending on the unobserved state of the world $\omega$. If the state of the world is $\omega =$ good, the research process yields a breakthrough at an exponential rate $\lambda k_i$ independently across agents, where $k_i$ is the measure of effort allocated by agent $i$ to research. If the state of the world is $\omega=$ bad, the research process never yields a breakthrough. All effort decisions and breakthrough events are observable to all participants, so as usual I assume that all agents share a common belief that the state of the world is good, which I denote $p(t)$ at time $t$.\footnote{This assumption is the standard ``no signaling what you don't know'' restriction.} In general, if any agent is choosing $k_i(t) > 0$ at time $t$, the agent is \textit{experimenting}. 

\paragraph{Breakthrough}
As stated, if a breakthrough arrives, the research game ends. At that instant, a lump-sum instantaneous reward arrives of size $R > 0$. The total continuation value of all agents improves to $\Pi > N\pi_s$.\footnote{I distinguish between these two objects primarily for continuity with previous literature; in the \cite{krc2005} model, $\Pi = N g$ and $R = h$.} Motivated by studying the payoff externalities imposed by possible sharing contracts, we will grant the agent making the breakthrough (``winner'') an instantaneous payoff of $R_w$ and a continuation payoff of $\pi_w$. The other agents receive an instantaneous reward $R_l$ and continuation payoff $\pi_l$. Since the total payoff is fixed, $\Pi = \pi_w + (N-1)\pi_l$ and $R = R_w + (N-1)R_l$.




\paragraph{Outcomes and Payoffs}
While the game has not ended, a history $h_t$ is given by a measurable path of effort choices, $\{ (k_1(s), k_2(s), \dots k_N(s)) \mid k_i(s) \in [0,1], \ s \le t \}$. 

An outcome of the game is a triple $\left(\tau, w, h_\tau\right)$. The first element $\tau \in \mathbb{R}_+ \cup \{ \infty \}$ is the realization of a stopping time, namely the stochastic arrival time of the breakthrough. 
Note that $\tau = \infty$ if the state of the world is bad or experimentation stops before breakthrough. 
The second element $w \in \{1, 2, \dots N\} \cup \{ \emptyset \}$ denotes the identity of the winner; if there is no winner ($\tau = \infty$), $w = \emptyset$. Finally, 
$h_\tau$ is the history of effort choices up to the stopping time. Note that $h_\tau$ implies paths $\{ k_i(t) \mid t \le \tau \}$ for every $i \in \{1, 2, \dots N \}$. 

Given some outcome $(\tau, w, h_\tau)$, one can formally define payoffs. All agents discount payoffs at a rate $r > 0$. If $\tau < \infty$, then the realized payoff to a winner $i$ (that is, $i = w$) is
\[ \Pi^W_i\left(\tau, w, h_\tau\right) = \int_0^\tau re^{-rt} \pi_s \left(1 - k_i(t)\right)\ \dd t + re^{-r\tau}R_w + e^{-r\tau} \pi_w. \]
The realized payoff to a loser $i \neq w$ is 
\[ \Pi^L_i\left(\tau, w, h_\tau\right) = \int_0^\tau re^{-rt} \pi_s \left(1 - k_i(t)\right)\ \dd t + re^{-r\tau}R_l + e^{-r\tau} \pi_l. \]
If the breakthrough never arrives $(\tau = \infty)$, the only payoffs come from the status quo technology, so the payoffs are given by 
\[  \Pi^{N}_i\left(h_{\infty}\right) = \int_0^\infty  re^{-rt} \pi_s \left(1 - k_i(t)\right)\ \dd t . \]
The structure and payoffs of the game are common knowledge to all agents. 

\paragraph{Beliefs and Strategies}
Fix a history $h_t$. History $h_t$ implies a realization of the path of total effort; let that path of total effort be $K(t) = \sum_i k_i(t)$. By assumption, the belief process $p(t)$ over the state of the world $\omega$ is public and common to all agents. To understand how beliefs evolve, consider an infinitesimal time increment $[t, t + \dd t)$. The flow probability of no breakthroughs occurring, conditional on $\omega = $ good, is $1 - K(t)  \lambda \dd t $. If $\omega = $ bad, breakthroughs never occur. Proceeding heuristically, the evolution of the belief according to Bayes' rule is given by
\[ p(t) + \dd p(t) = \frac{p(t)(1 - K(t)\lambda \dd t) }{(1-p(t)) + p(t)(1 - K(t)\lambda \dd t)}. \]
Rearranging and taking the limit as $dt \to 0$, the belief evolves as 
\begin{equation}\label{eqn:belief_law}
\dd p(t) = - K(t) \lambda p(t)(1-p(t)) \dd t.
\end{equation}
Note that the sign of $dp(t)$ is nonpositive; that is, if $K(t) > 0$, the belief becomes more pessimistic in the absence of the game-ending breakthrough. If $\tau$ is finite, then the belief $p$ jumps to 1 at $\tau$.

Since the cooperative problem of a planner maximizing the joint payoffs of all agents is a Markov decision process over the belief $p$, there exists an optimal first-best policy that is measurable with respect to $p$, which I explicitly solve for in Section \ref{sec:cooperative}. Motivated by this, I also focus on Markov strategies with respect to the public belief in the noncooperative game. Formally, in the noncooperative game, the strategy of agent $i$ is a map $\sigma_i : [0,1] \to [0,1]$ from the public belief into an effort level. As usual, denote a profile of Markov strategies for all agents except $i$ as $\sigma_{-i}$. Given the focus on Markov strategies, I use Markov perfect equilibrium (MPE) as the solution concept. More precisely, an MPE is a profile $\{ \sigma_i \}_i$ such that at any state of the game, agent $i$'s strategy $\sigma_i$ is a best response policy to the other agents' strategies $\sigma_{-i}$, given the belief law of motion in (\ref{eqn:belief_law}). Note that for any fixed Markov $\sigma_{-i}$, the best response problem of agent $i$ is a Markov decision process, and so there always exists a best response policy that is Markov; hence, such an MPE is also an equilibrium even when a larger class of strategies is permitted. 

Furthermore, all agents, whether cooperatively or not, make effort decisions based on their subjective expectation over the distribution of future outcomes conditional on the history, and the evolution of the public belief reflects the learning process about the state of the world. That is, $K$ is not an exogenous process, but rather depends on the history of $p$ and $K$. Thus, for the stochastic belief process $p$ and action profiles $(k_1, k_2, ... k_N)$ to be well defined, I impose some further restrictions on strategies of the agents in the noncooperative game; namely, I focus on strategies $\sigma_i$ that are finite piecewise Lipschitz and left-continuous.\footnote{These conditions were originally introduced in \cite{kr2010}. The piecewise Lipschitz and left-continuous restriction is a technical assumption that ensures that the belief law of motion in \eqref{eqn:belief_law} is well defined. The strengthening to finite piecewise Lipschitz eliminates asymmetric infinitely switching equilibria in \cite{krc2005}. These equilibria do exist for certain choices of the game parameters, but I believe it is reasonable to eliminate these equilibria because they do not arise as limits of PBEs of discrete-time games and are an artifact of continuous time (see \cite{hkr2022}).} 

For this paper, I focus on symmetric equilibria; that is, all agents use a common continuation strategy after any history. This focus is a natural choice since the agents are ex ante symmetric (in the next section, I show that in the first-best solution, the optimal policy is symmetric across all agents). Some of the results can be strengthened to characterize asymmetric equilibria; since these results are auxiliary to the focus of the paper, I address these extensions in the appendix.

\subsection{Motivating Examples}
Having presented the model formulation for the baseline research game, I provide a few motivating examples.

\paragraph{Industry Research:} Several firms are engaged in researching a technological advance. The first firm to make the technology work can claim a patent, which alters its future revenue stream and that of its competitors. The extent to which the payoffs of the discovering firm's rivals benefit or suffer depends on how strictly the patent can be enforced. 

\paragraph{Contest Participation:} There are individuals engaged in a contest to develop a product. The first individual to succeed earns some prize, and the losers can also be compensated or penalized. 

\paragraph{Academic Collaboration:} Researchers are collaborating on a project. However, depending on their contribution to a project, they might receive different levels of credit (e.g. first authorship). The first author gains a different amount of credit than the subsequent authors, and one might like to know whether the discrepancy between the credit received by a first author results in first-best research in equilibrium.

\section{Cooperative Solution}
\label{sec:cooperative}
I now discuss the first-best solution, that is, the solution that maximizes the total payoff of all agents. I refer to this as the first-best or the efficient solution, and this is the benchmark for efficiency against which equilibria of the noncooperative game are measured.

Since the cooperative social planner can set the effort decisions of all agents, the problem is a continuous-time Markov decision process (MDP) over the state (which is the public belief $p$), and so an optimal policy exists among those that are measurable with respect to the state $p$. This problem is standard, so my discussion here is brief.\footnote{See \cite{krc2005} for a more in-depth discussion of the methodology}

Recall that $R$ and $\Pi$ are the total instantaneous and continuation payoffs after a breakthrough. The average value function of the agents over the state satisfies a Hamilton-Jacobi-Bellman (HJB) equation, which is given by
\begin{equation}\label{eqn:hjb_cooperative}
    V_N(p) = \pi_s + \max_K \left[K\left(p\frac{\lambda}{r}\left(\frac{\Pi}{N} - V_N(p) - (1-p)V_N'(p) \right) - \frac{c(p)}{N} \right) \right] 
\end{equation}
where
\begin{equation}\label{eqn:def_c}
    c(p) = \pi_s - p \lambda R.
\end{equation} 

Intuitively, the term 
$p\frac{\lambda}{r}\left(\frac{\Pi}{N} - V_N(p) - (1-p)V_N'(p) \right)$
denotes the flow average benefit of experimentation, which is the average increase in continuation payoffs (from $V_N(p)$ to $\frac{\Pi}{N}$) 
minus the downward effect on payoffs from becoming more pessimistic about the state of the world $(1-p)V_N'(p)$. The cost $c(p)$ denotes the myopic opportunity cost of experimentation. 

I explicitly solve for the value function satisfying this HJB equation and obtain the following result characterizing the efficient solution. 
\begin{theorem}\label{thm:cooperative}
The first-best solution has all agents exerting full effort $k_i = 1$ on the research project until the public belief reaches the first-best threshold 
\begin{equation}\label{eqn:belief_fb}
    p_{FB} = \frac{\pi_s}{\lambda R + \frac{\lambda}{r}\left(\Pi- N\pi_s \right) },
\end{equation}
and no effort is exerted on research after the belief falls below $p_{FB}$. 
\end{theorem}

A few features of the efficient solution are worth emphasizing. First, the planner implements symmetric strategies, so the cooperative first-best policy for each agent is identical. Another consequence of the solution is that the optimal policy is ``bang-bang'' at a cutoff; the cooperative planner either sets total effort $N$ into the research project or nothing, depending on whether the belief is above or below $p_{FB}$. Hence, implementation of the first-best requires that all agents exert full effort into research above $p_{FB}$ and drop the research project below $p_{FB}$. 

Formally, I call a Markov strategy $\sigma_i$ a \textbf{cutoff} strategy if $\sigma_i(p) = 1$ for $p > p_T$, and $\sigma_i(p) = 0$ for $p \le p_T$ for some $p_T$. I refer to $p_T$ as the \textbf{threshold} belief. The efficient solution consists of cutoff strategies with threshold $p_{FB}$ defined in (\ref{eqn:belief_fb}).

\section{Noncooperative Game}
\label{sec:noncoop}
Having characterized the efficient solution, I continue onto the analysis of the noncooperative game.  The first main result shows that the nature of the equilibria of the noncooperative game critically depends on a sharp condition on the game parameters concerning the payoffs of the losers.

\begin{theorem}\label{thm:noncoop_efficient}
    The efficient solution is an MPE of the noncooperative game if and only if $\frac{\pi_s - \pi_l}{r} = R_l $. Furthermore, if $\frac{\pi_s - \pi_l}{r} = R_l $, the efficient solution is also the unique MPE. 
\end{theorem}

Note that the condition is if-and-only-if and hence sharply characterizes whether the efficient solution is an MPE. Furthermore, the second part of the statement implies unique implementation; that is, the efficient solution is the only MPE when the condition holds, and in particular this implies there are no other asymmetric or nonmonotone equilibria (which can arise when the efficiency condition fails). Note that the condition is an equality condition (knife-edge) on parameters of the model, and so equilibria are generically inefficient. 

As a brief aside, I can characterize the nature of the weakly monotonic symmetric MPEs when the efficiency condition fails.  Intuitively, when $\frac{\pi_s - \pi_l}{r} > R_l$, breakthroughs harm the losers; in the noncooperative game, agents overexperiment due to the incentive to not lose. When $\frac{\pi_s - \pi_l}{r} < R_l$, breakthroughs benefit the losers; hence, the incentive to free-ride induces underexperimentation. When $\frac{\pi_s - \pi_l}{r} = R_l$, breakthroughs are neutral for the losers, and the noncooperative outcome is efficient. Since this paper focuses on efficiency and contracts that restore efficiency, I relegate the formal statements to the Appendix, which contains the general characterization of MPEs when the efficiency condition fails. 

\subsection{Discussion}
Some remarks on the efficiency condition
\[ \frac{\pi_s - \pi_l}{r} =  R_l \]
 are in order. The left-hand side is the difference in present discounted value of the status quo technology and the loser technology, and the right-hand side is the lump-sum compensatory reward that a loser receives at time of breakthrough. Hence, the economic interpretation of the condition is whether the lump-sum reward a loser receives compensates for the change in the technology value to the loser.

The efficiency condition has no dependence on the number of agents $N$ or the arrival rate of the breakthrough process $\lambda$. That is, if a designer were selecting game parameters to attempt to implement efficiency as an equilibrium of the noncooperative game, increasing/decreasing the size of the winner's rewards (so long as $\pi_w \ge \pi_l$) does not impact the efficiency of the result. Altering the number of agents $N$ and the breakthrough rate $\lambda$ also does not change the nature of the equilibria in terms of efficiency. Increasing $N$ does not change whether the equilibrium is efficient, but exacerbates any existing inefficiency (i.e., the difference in belief between the cutoff belief for the first-best and the cutoff in any MPE). Changing $\lambda$ scales the noncooperative game and the first-best solution identically and hence also has no impact on efficiency.

More notably, the efficiency condition is independent of $R_w$ and $\pi_w$; that is, the condition for efficiency does not depend on what the winner receives. To understand this, consider an incremental increase in $R_w$ or $\pi_w$, fixing $R_l$ and $\pi_l$. Any incremental increase induces agents in the noncooperative game to experiment more but also induces more experimentation in the social planner problem. To understand why this does not affect the efficiency implications, consider the best-response problem of agent $i$. For agent $i$, the cumulative effort of other agents $K_{-i}$ brings about a ``loss'' event at a rate of $\lambda K_{-i}$, but agent $i$ has no agency over the effort of the other agents. Instead, the effort choice for agent $i$ weighs the relative benefit of winning versus the status quo project (and in fact, $b_I$ and $c_I$ have no $R_l, \pi_l$ dependence). Hence, in the noncooperative game, agents are intuitively already trading off the relative benefits of winning versus the status quo in an efficient manner; however, they do not properly account for the externalities induced by their research effort on other players. Thus, the source of any potential inefficiency is the extent to which the externalities harm or help the other agents. 

The next subsection outlines the technical steps to prove Theorem \ref{thm:noncoop_efficient}. A reader less concerned with these details may skip the next subsection and proceed to Section \ref{sec:contracts}.

\subsection{Best Response Problem}
Since I am interested in Markov perfect equilibria, I start by considering the best response problem of a single agent reacting to a given profile of Markov strategies of the other agents. That is, suppose that the strategies of all other agents are exogenously fixed at $\{ \sigma_j \}_{j \neq i}$. The cumulative effort of the other agents at some belief $p$ is $K_{-i}(p) = \sum_{j \neq i} \sigma_j(p)$, which is also left-continuous and finite piecewise Lipschitz by assumption.
Using standard arguments, I derive the HJB equation characterizing best-response $k_i$ given the function $K_{-i}(p)$. Let $u(\cdot)$ denote the value function of agent $i$'s best-response Markov decision process. Then the HJB equation characterizing $u$ is given by 
\begin{equation}\label{eqn:hjb_baseline}
    u(p) = \pi_s + K_{-i}(p) \left(p\lambda R_l + b_I(p,u,u') - p\frac{\lambda}{r}(\pi_w - \pi_l) \right)+  \max_{k_i} \left[k_i \left( b_I(p, u, u') - c_I(p) \right)  \right],
\end{equation}
where
\begin{equation}\label{eqn:def_bI}
     b_I(p, u, u') = p\frac{\lambda}{r}(\pi_w - u(p) - (1-p)u'(p)),
\end{equation}  
\begin{equation}\label{eqn:def_cI}
     c_I(p) = \pi_s - p \lambda R_w.
\end{equation}  

\begin{lemma}\label{lem:viscosity}
    Fixing $K_{-i}(p)$, there is a unique viscosity solution of \eqref{eqn:hjb_baseline}. 
\end{lemma}
The proof is in the Appendix, but intuitively requires checking that the problem is sufficiently well-behaved (compact action space, continuous and bounded law of motion that is Lipschitz in the state, and Lipschitz continuous payoffs in actions). Note that the generality of this setting implies that viscosity solutions are \textit{necessary} here, and the notion of a viscosity solution is not just an exercise for the sake of generality.\footnote{It happens that much of the literature can sidestep this concern because the parameter values work out to admit differentiable solutions, but for the generality I consider, nondifferentiable solutions are unavoidable.} That is, for a range of parameter values, there exists no differentiable solution to \eqref{eqn:hjb_baseline}. In those cases, a viscosity solution always exists, which is differentiable almost everywhere and satisfies the HJB equation when it is differentiable, with additional constraints at kinks.\footnote{These kinks occur when the experimentation game features overexperimentation; for a full characterization, see the Appendix. In such games, kinks occur precisely when experimentation stops, which is also the case in \cite{kr2015}; in this case, it is due to discontinuity in the strategies of the other players.}

I briefly discuss the economic intuition behind the HJB equation in \eqref{eqn:hjb_baseline}. The term $b_I$ defined in \eqref{eqn:def_bI} denotes the individual perceived continuation benefit of experimentation; that is, if the state of the world is good (which has perceived probability $p$), the agent generates a breakthrough and wins the contest at rate $\lambda$. The breakthrough generates a shift in continuation payoffs from $u(p)$ to $\pi_w$, and the final $(1-p)u'(p)$ term denotes the marginal downward effect on payoffs from continued experimentation if no breakthrough arrives. The $c_I$ expression defined in \eqref{eqn:def_cI} denotes the myopic opportunity cost of experimentation. Note also the presence of the externality term of others' effort; $K_{-i}$ brings about losing at perceived rate $p \lambda$, which yields a lump-sum reward $R_l$ but also the continuation benefit $b_I$, but under the losing state (hence subtracting the $\pi_w - \pi_l$ term). 

It is immediate from \eqref{eqn:hjb_baseline} that agent $i$ finds it optimal to take $k_i = 1$ if $b_I > c_I$, $k_i = 0$ if $b_I < c_I$, and any action if $b_I = c_I$. However, $b_I$ is a complex mathematical object, since it depends both on the value function $u$ as well as the derivative of the value function, $u'$. However, a standard argument allows me to simplify the best-response policy such that in any MPE, the best-response policy must satisfy a simpler condition that eliminates the dependence on $u'$.\footnote{This technique originally appeared in \cite{bh1999}.}

\begin{lemma}\label{lem:baseline_br}
The best-response policy in any equilibrium satisfies:
\begin{equation}\label{eqn:baseline_br}
    k_i = \begin{cases}
    0 & u(p) < \pi_s + K_{-i}(p)\left(\pi_s - p\lambda (R_w - R_l) - \frac{p\lambda}{r}\left( \pi_w - \pi_l \right) \right) \\
    \in [0,1] & u(p) = \pi_s + K_{-i}(p)\left(\pi_s - p\lambda (R_w - R_l)- \frac{p\lambda}{r}\left( \pi_w - \pi_l \right) \right)  \\
    1 & u(p) > \pi_s + K_{-i}(p)\left(\pi_s - p\lambda (R_w - R_l) - \frac{p\lambda}{r}\left( \pi_w - \pi_l \right) \right)  \\
    \end{cases}
\end{equation}
\end{lemma}

A key simplification made by Lemma \ref{lem:baseline_br} is that the conditions for the best-response policy now depend only on $u$ and $p$ (and not $u'$). That is, as a thought experiment, consider the best-response policies when $K_{-i}(p)$ is fixed at a constant level $K_{-i}$. Define the level curves for each $K_{-i}$, 
\[ \mathcal{D}_{K_{-i}} = \left\{(p,u) \in [0,1]\times \mathbb{R}_+ \left\vert \, u = \pi_s + K_{-i}\left(\pi_s - p\lambda (R_w - R_l) - \frac{p\lambda}{r}\left( \pi_w - \pi_l \right) \right)\right.\right\} . \] 
If the belief is $p$, Lemma \ref{lem:baseline_br} implies that the best-response policy should be to exert effort when the current state $(p, u(p))$ lies above or below the level curve $\mathcal{D}_{K_{-i}}$. Note that $\mathcal{D}_0$ is a flat line $u = \pi_s$, and $\mathcal{D}_{K_{-i}}$ is a downward-sloping line for $K_{-i} > 0$. Note further that all the $\mathcal{D}_{K_{-i}}$ intersect at the same point in the $(u, p)$ plane, when $u = \pi_s$ and at belief $p_\times$ characterized by:
\[ \pi_s = p_\times \left(\lambda (R_w - R_l) + \frac{\lambda}{r}\left(\pi_w - \pi_l \right) \right)  \]
\begin{equation}\label{eqn:belief_cross}
    p_\times = \frac{\pi_s}{\lambda (R_w - R_l) + \frac{\lambda}{r}\left(\pi_w - \pi_l \right)}.
\end{equation}

Since each of the level curves $\mathcal{D}_{K_{-i}}$ intersect at $p_\times$, one can informally see that if $p_\times > p_{FB}$ the first-best solution is not an equilibrium of the noncooperative game; intuitively, this is because the point $(p_{FB}, \pi_s)$ lies below $\mathcal{D}_{N-1}$, and so Lemma \ref{lem:baseline_br} implies that the best-response policy of some agent falls into the first category, and so someone should have stopped experimenting earlier.

It turns out that the characterization of whether $p_\times$ is larger or smaller than $p_{FB}$ depends on a simple condition:\footnote{A degenerate case arises if $R_w = R_l$ and $\pi_w = \pi_l$ (i.e., the ``collaborating'' case from \cite{bh2011}). In this case, set $p_\times = \infty$, as the level curves $\mathcal{D}_{K_{-i}}$ are parallel.)} 

\begin{lemma}
\label{lem:ptimes_pFB_comparison}
    If $\frac{\pi_s - \pi_l}{r} > R_l$, then $p_\times < p_{FB}$. If $\frac{\pi_s - \pi_l}{r} < R_l$, then $p_\times > p_{FB}$. If $\frac{\pi_s - \pi_l}{r} = R_l$, $p_\times = p_{FB}$.
\end{lemma}
\begin{proof}
Note that 
\[ \pi_w - \pi_l = \pi_w + (N-1)\pi_l - N\pi_s + N(\pi_s - \pi_l)\]
and
\[ R_w - R_l = R_w + (N-1)R_l - NR_l = R - NR_l \le R . \]
Using this, we can rewrite Equation \eqref{eqn:belief_cross} as
\begin{align*} p_\times &= \frac{\pi_s}{\lambda R - N\lambda R_l + \frac{\lambda}{r}\left(\pi_w + (N-1)\pi_l - N\pi_s + N(\pi_s - \pi_l)\right)}\\
&= \frac{\pi_s}{\lambda R  + \frac{\lambda}{r}\left(\Pi- N\pi_s\right) + N\lambda\left( \frac{\pi_s - \pi_l}{r} - R_l \right)}.
\end{align*}
Thus, whether $p_\times$ is larger or smaller than $p_{FB}$ exactly depends on the sign of the last term in the denominator, 
\[ N\lambda\left( \frac{\pi_s - \pi_l}{r} - R_l \right). \]
Checking each case of the parenthesized quantity thus yields the result.
\end{proof}

Another important benchmark to consider is the belief where experimentation would stop if $K_{-i} = 0$ (that is, no other agents were experimenting). If $K_{-i}$ was fixed to zero, then the best-response of a single agent is just an optimal control problem with HJB equation
\begin{equation}\label{eqn:indiv_hjb}
    u(p) = \pi_s + \max_{k_i} \left[k_i \left( p\frac{\lambda}{r}(\pi_w - u(p) - (1-p)u'(p))- c_I(p) \right)  \right],
\end{equation}  
or equivalent to the cooperative HJB equation with one agent, Equation \eqref{eqn:hjb_cooperative} with $N=1$ and total instantaneous payoff $R_w$. By Theorem \ref{thm:cooperative}, the best response is then a cutoff strategy with threshold belief 
\begin{equation}\label{eqn:belief_indiv}
    p_I = \frac{\pi_s}{\lambda R_w + \frac{\lambda}{r}(\pi_w - \pi_s)} 
\end{equation} 
Intuitively, this belief quantifies the individual incentive to exert effort on research; if $p_{I} < p_{FB}$, intuitively the first-best solution cannot be sustained because some agent has an incentive to continue working on research at $p_{FB}$ if every other agent stops.

Remarkably, if $p_I \ge p_{FB}$, then $p_\times \ge p_{FB}$, and if $p_I < p_{FB}$ then $p_\times < p_{FB}$. In fact, the following stronger characterization is true:
\begin{lemma}\label{lem:pI_characterization}
    The cutoff $p_I$ lies between $p_{FB}$ and $p_\times$.
\end{lemma}
\begin{proof}
We can rewrite this similarly as we did with $p_\times$:
\begin{align*}
    p_I &= \frac{\pi_s}{\lambda R_w + \frac{\lambda}{r}(\pi_w - \pi_s)} \\
    &= \frac{\pi_s}{\lambda R + \frac{\lambda}{r}(\pi_w + (N-1)\pi_l - N\pi_s) + (N-1) \frac{\lambda}{r}\left( \pi_s - (rR_l + \pi_l) \right)}  \\
    &= \frac{\pi_s}{\lambda R + \frac{\lambda}{r}(\Pi - N\pi_s) + (N-1) \lambda \left(  \frac{\pi_s - \pi_l}{r} - R_l\right)} .
\end{align*}

Note that the denominator is almost exactly the same as $p_\times$, except that the last term in the denominator has a coefficient $(N-1)$ instead of $N$. 
In other words, the denominator of $p_I$ is the weighted average of the denominators of $p_\times$ and $p_{FB}$.
Hence, if $\frac{\pi_s - \pi_l}{r} > R_l$, $p_I \in [p_\times, p_{FB}]$, and if $\frac{\pi_s - \pi_l}{r} < R_l$, $p_I \in [p_\times, p_{FB}]$. If $\frac{\pi_s - \pi_l}{r} = R_l$, $p_I = p_\times = p_{FB}$. In any case, $p_I$ always lies between $p_\times$ and $p_{FB}$.
\end{proof}

Now, I formally show that $p_\times$ and $p_I$ provide bounds on the end of experimentation in noncooperative equilibria. 

First, $p_I$ provides an upper bound on the end of experimentation under certain conditions. That is, define a \textit{weakly monotonic} strategy as one where the strategy is weakly monotonic in the belief (and note that cutoff strategies are weakly monotonic). Then the following holds:
\begin{lemma}\label{lem:stopping_upper_bound}
    Suppose that in some MPE, all agents use weakly monotonic strategies. Then experimentation cannot stop at any $p > p_I$.
\end{lemma}
\begin{proof}
    Suppose, for sake of contradiction, that experimentation stops at some $p_T > p_I$ for some MPE in weakly monotonic strategies. Since agents stopped experimenting at $p_T$, they exert zero effort at any $p < p_T$ since the equilibrium is in weakly monotonic strategies, and hence the value function $u(p) = \pi_s$ for $p < p_T$ by the boundary condition. Consider the best-response problem of an arbitrary agent $i$. By the boundary conditions and the HJB equation, at any point $p \in [p_I, p_T]$, the HJB equation indicates that
    \begin{align*} 
    \pi_s = &\pi_s + \max_{k_i} \left[k_i \left( p \frac{\lambda}{r}(\pi_w - \pi_s) - (\pi_s - p \lambda R_w) \right)  \right] \\
    0 =&  \max_{k_i} \left[k_i \left( p\frac{\pi_s}{p_I}  - \pi_s \right)  \right].
    \end{align*}
However, this is a contradiction; since $p > p_I$, the maximal $k_i$ is 1, and hence the RHS here cannot be zero, but $p_T > p$. Hence, experimentation cannot stop at $p_T > p_I$.

\end{proof}

Now, I show that $p_\times$ can provide a lower bound on the end of experimentation when $\frac{\pi_s - \pi_l}{r} \ge R_l $.
\begin{lemma}\label{lem:stopping_lower_bound}
    Suppose that $\frac{\pi_s - \pi_l}{r} \ge R_l $. Then in any MPE, experimentation must stop at some $p \ge p_\times$.
\end{lemma}
\begin{proof}
Suppose for sake of contradiction that experimentation stops at some $p_T < p_\times$. Then, some agent was exerting a positive amount of effort at beliefs down to $p_T$. Let $u$ be the value function of that agent. Since experimentation stops at $p_T$, $u(p_T) = \pi_s$. Since $u$ must be a viscosity solution to the best-response HJB equation for some $K_{-i}(p)$, we can take a sequence $p_n \to p_T$ such that $u'(p_n)$ is well defined, $p_n > p_T$. Since the agent was exerting a positive amount of effort, $k_i(p_n) > 0$, so it must have been the case that 
\[ p_n\frac{\lambda}{r} \left(\pi_w - u(p_n) - (1-p_n)u'(p_n)\right) > \pi_s - p_n \lambda R_w\]
\[ p_n \left(\lambda R_w + \frac{\lambda}{r} \left(\pi_w - u(p_n) - (1-p_n)u'(p_n)\right)\right) > \pi_s\]
\[ p_n \left(\frac{\pi_s}{p_I}+ \frac{\lambda}{r} \left(\pi_s - u(p_n) - (1-p_n)u'(p_n)\right)\right) > \pi_s\]
\[ p_n \left(\frac{\lambda}{r} \left(\pi_s - u(p_n) - (1-p_n)u'(p_n)\right)\right) > \pi_s\left(1 - \frac{p_n}{p_I} \right).\]
Taking the limit as $p_n \to p_T$,
\[ - \frac{\lambda}{r} p_T(1-p_T)u'_+(p_T) > \pi_s\left(1 - \frac{p_T}{p_I} \right),\]
where $u'_+$ denotes the right derivative, since $u$ need not be differentiable at $p_T$. The right-hand side is positive because $\frac{\pi_s - \pi_l}{r} \ge R_l $ implies that $p_T < p_\times \le p_I$. Therefore, it must be the case that $u'_+(p_T) < 0$. Hence, there must be some point $p \in [p_T, p_\times]$ such that $u(p) < \pi_s$. However, this implies that the point $p, u(p)$ lies below $\mathcal{D}_{K_{-i}}$ for all $K_{-i} \in [0, N-1]$ (since every $\mathcal{D}$ has a nonpositive slope passing through $(p_\times, \pi_s)$) and an equilibrium exists where some agent is exerting a positive measure of effort at that point, a contradiction of Lemma \ref{lem:baseline_br}. 
\end{proof}

Now, I present the proof intuition (the formal proof is left to the Appendix). 

\paragraph{Proof Intuition} To show that the efficient solution is an MPE if $\frac{\pi_s - \pi_l}{r} = R_l$, it suffices to check that the average value function solving the first-best solution HJB equation also solves the best-response HJB equation when all other agents play cutoff strategies at $p_T$. 
To show the converse, note that if $\frac{\pi_s - \pi_l}{r} < R_l$, Lemma \ref{lem:ptimes_pFB_comparison} implies that $p_\times > p_{FB}$, so the point $(p_{FB}, \pi_s)$ lies strictly interior in the half plane below $\mathcal{D}_{N-1}$. (See Figure \ref{fig:ptimes} for an illustration.) Thus, playing a cutoff strategy at $p_{FB}$ cannot be a best response by Lemma \ref{lem:baseline_br}, since it would imply $k_i = 1$ below $\mathcal{D}_{N-1}$. In the other case, if $\frac{\pi_s - \pi_l}{r} > R_l$, Lemmas \ref{lem:ptimes_pFB_comparison} and \ref{lem:pI_characterization} imply that $p_I < p_{FB}$. Thus, the efficient solution cannot be an MPE, since it would imply that there exists an MPE where agents use weakly monotonic strategies but experimentation stops at $p_{FB} > p_I$, a contradiction of Lemma \ref{lem:stopping_lower_bound}. Together, this implies that the efficient solution is an MPE iff $\frac{\pi_s - \pi_l}{r} = R_l$.
Finally, to show uniqueness, note that Lemma \ref{lem:stopping_upper_bound} shows that if $\frac{\pi_s - \pi_l}{r} = R_l$, experimentation must stop at or above $p_\times = p_{FB}$. I then show that regardless of what the other agents do, the best response for an agent is to use a cutoff strategy at $p_{FB}$, and hence the only MPE must be the efficient solution.

\begin{figure}
    \centering
         \begin{tikzpicture}[scale=0.8]
            \draw[->, thick] (-1,0) -- (5,0) node[anchor=west]{$p$};
            \draw[->, thick] (0,-1,0) -- (0,5) node[anchor=south]{$u$};
            \draw[dotted] (1,5) -- (1,0) node[anchor=north]{$p_{FB}$};
            \draw[teal] (2,1.5) -- (0,1.5) node[anchor=east]{$\pi_s$};
            \draw[teal, dashed] (1,1.5) .. controls (2,1.5) and (3,2) .. (5,5) node[anchor=south]{$V_{FB}$};
            \draw (0,1.5) -- (5,1.5) node[anchor=west]{$\mathcal{D}_0$};
            \draw (0,4) -- (4.5,0.7) node[anchor=west]{$\mathcal{D}_{N-1}$};
            \draw[dotted] (3.41,5) -- (3.41,0) node[anchor=north]{$p_\times$};
        \end{tikzpicture}
    \caption{Intuitive depiction for why the first-best solution cannot be an equilibrium of the noncooperative game if $p_\times > p_{FB}$. $V_{FB}$ is the value function corresponding to the first-best solution. Note that Lemma \ref{lem:baseline_br} implies that it cannot be an equilibrium best response to continue experimenting once $V_{FB}$ falls below $\mathcal{D}_{N-1}$.}
    \label{fig:ptimes}
\end{figure}
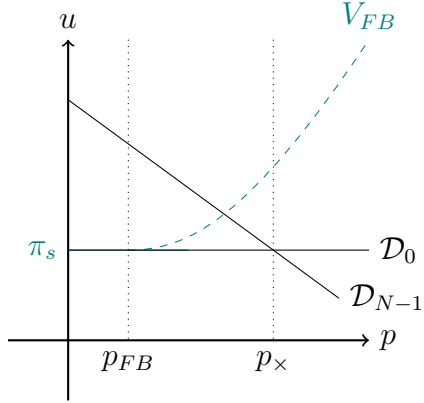

\section{Sharing Contracts}
\label{sec:contracts}
A logical way the agents might wish restore cooperative efficiency is if they can agree ex ante to a contract that specifies how to split the rewards from experimentation in the event of a breakthrough. Thus, in this section, I consider the problem of a regulator (or contest designer) who observes the outcome of the baseline experimentation game and decides how to award payoffs. I first formalize the broad mathematical definition of a sharing contract. I then show that within this broad class of contracts, efficiency can be restored by very simple contracts that only require the regulator to observe \textit{part} of the outcome. More concretely, a regulator can restore efficiency if the regulator observes the winner/losers, or if the regulator observes the profile of effort at the end of the game. Notably, for the regulator, either piece of information is sufficient to restore efficiency. That is, the regulator does not need to observe the full history of effort.

\subsection{Contract Formalism}
Recall that an outcome of the baseline experimentation game is a triple $(\tau, w, h_\tau)$, where $\tau\in \mathbb{R}_+ \cup \{ \infty \}$ is the stopping time corresponding to the arrival of the breakthrough, $w \in \{1, 2, ... N\}$ is the winner, and $h_\tau$ is the history of effort paths leading up to time $\tau$. Let $\mathcal{O}$ be the space of outcomes. Define a \textit{sharing contract} as a map $c : \mathcal{O} \to \mathbb{R}_+^{N} \times \mathbb{R}_+^N$ such that for $c(\tau, w, h_\tau) = ((R_1, ...R_N) , (\pi_1, ... \pi_N))$, then $\sum_{i}  R_i = R$ and $\sum_{i}  \pi_i = \Pi$. Intuitively, the sharing contracts map outcomes of the experimentation game into instantaneous and continuation payoffs for the agents, fixing the total instantaneous payoff at $R$ and the total continuation payoff at $\Pi$. 

The timing is then as follows: before the experimentation game is played, the regulator commits to a sharing contract $c$. The agents play the experimentation game in the baseline model. Given an outcome $(\tau, w, h_\tau)$ let $c(\tau, w, h_\tau) = ((R_1, ...R_N) , (\pi_1, ... \pi_N))$. Then agent $i$ receives $R_i$ instantaneous reward at the time of the breakthrough and $\pi_i$ of the continuation reward.

Note that although I define contracts as maps from the full outcome to payoffs, I show in the next two sections that simple contracts that map on much less information can restore efficiency. 

\subsection{Contracting on Winner/Losers}
First, I show that to restore efficiency, it is sufficient for sharing contracts to only condition on $w$; further, within this class of contracts, it suffices to focus on simple linear contracts. In particular, I consider linear sharing contracts that provide the winner a fraction $\alpha_I$ of the instantaneous reward and $\alpha_C$ of the continuation reward and split the remainder of the rewards evenly among all participants. That is, formally, 
\[ c^W_{\alpha_I, \alpha_C}(w) =  \left(\left\{ \left(\alpha_I \mathbbm{1}[w = i] + \frac{1-\alpha_I}{N}\right) \cdot R\right\}_i, \left\{ \left(\alpha_C \mathbbm{1}[w = i] + \frac{1-\alpha_C}{N}\right) \cdot \Pi \right\}_i\right) \]
where the superscript $W$ refers to the fact that these contracts condition on the observation of $w$. Note that depending on the values of $\alpha_I, \alpha_C$, the contract might ex post result in the losers making payments to the winner or the winner compensating the losers.

Recall that in the case $\frac{\pi_s - \pi_l}{r} = R_l$, Theorem \ref{thm:noncoop_efficient} shows that the unique competitive MPE outcome is the efficient solution. To that end, define the guarantee of a $c^W_{\alpha_I, \alpha_C}$ contract as

\begin{equation}\label{eqn:def_guarantee}
    G\left(c^W_{\alpha_I, \alpha_C}\right) := rR\left(\frac{1-\alpha_I}{N} \right) + \Pi\left(\frac{1-\alpha_C}{N}\right) .
\end{equation} 

Intuitively, this is the expected flow value of a loser upon a breakthrough. The following result follows immediately from Theorem  \ref{thm:noncoop_efficient} and characterizes efficient contracts in terms of the guarantee of the contract: 

\begin{theorem}\label{thm:guarantee}
The first-best solution is an MPE of the game under sharing contract $c^W_{\alpha_I, \alpha_C}$ if and only if $G\left(c^W_{\alpha_I, \alpha_C}\right)  = \pi_s$. Furthermore, if $G\left(c^W_{\alpha_I, \alpha_C}\right)  = \pi_s$, then the first-best solution is the unique MPE. 
\end{theorem}
\begin{proof}
    Note that contract $c^W_{\alpha_I, \alpha_C}$ induces an experimentation game with 
    \[ \tilde{R}_w = \alpha_I R + (1-\alpha_I) \frac{R}{N}, \quad \tilde{R}_l = (1-\alpha_I) \frac{R}{N} \]
    \[ \tilde{\pi}_w = \alpha_C \Pi  + (1-\alpha_C) \frac{ \pi_w + (N-1)\pi_l}{N}, \quad\tilde{\pi}_l = (1-\alpha_C) \frac{ \Pi }{N} \]
    Then,
    \[ \tilde{\pi}_l + r\tilde{R}_l = G\left(c^W_{\alpha_I, \alpha_C}\right) \]
    Thus, 
    \[ G\left(c^W_{\alpha_I, \alpha_C}\right)  = \pi_s \iff  \frac{\pi_s - \tilde{\pi}_l}{r} = \tilde{R}_l  \]
    so by Theorem \ref{thm:noncoop_efficient}, the result follows.
\end{proof}

Note that $G\left(c^W_{\alpha_I, \alpha_C}\right) = \pi_s$ is satisfied by many sharing contracts, and $\alpha_I$ and $\alpha_C$ are substitutable parameters. These parameters control for how competitive each part of the reward is, and hence it makes sense that they alter incentives in the same way. In particular, one particular contract does not even require the regulator to know $\lambda$ or $r$: specifically, $\alpha_I = 1$ and $\alpha_C = 1 - \frac{N\pi_s}{\Pi}$ induces a contract with guarantee equal to $\pi_s$ and has no $r$ dependence.

From examining the expression for the guarantee, it is immediate that sharing the instantaneous reward and sharing the continuation benefits are complementary instruments to restore efficiency. 

\subsubsection{Unobservable Actions}
Suppose, in the baseline game, that agents cannot observe each others' actions (effort decisions) but the identity of a winner is still observable. In this setting, there is no longer a public belief, and so the strategies of the agents are functions of time only. I show that the same efficiency condition as Theorem \ref{thm:noncoop_efficient} still holds. That is, recall the first-best cutoff belief defined in \eqref{eqn:belief_fb} and the belief law of motion \eqref{eqn:belief_law}. Given that the first-best solution has all agents exerting full effort on research until the belief reaches $p_{FB}$, let $t_{FB}$ denote the time when (in the absence of a breakthrough) all agents stop working on research in the first-best solution. I can explicitly solve for the path of the belief in the first-best, and the characterization of $t_{FB}$ satisfies 
\begin{corollary}\label{corr:unobservable_fb}
In the first-best with unobservable actions, all agents exert full effort until $t_{FB}$ if no breakthrough occurs and stop after, where $t_{FB}$ is given by
\begin{equation}
    \label{eqn:fb_time}
    t_{FB} = \frac{1}{N\lambda} \ln \left( \frac{\Omega(p_{FB}) }{\Omega(p(0))} \right) ,
\end{equation}
where $\Omega$ defines the odds ratio function: 
\begin{equation}\label{eqn:odds_ratio} 
\Omega(p) := \frac{1-p}{p}  .
\end{equation}
\end{corollary}
The proof is left to the Appendix; it entails solving the belief law of motion in the first-best solution characterized by Theorem \ref{thm:cooperative}. 

\begin{proposition}\label{prop:unobservable}
    When actions are unobservable, all firms using individual cutoff strategies until time $t_{FB}$ is an equilibrium of the noncooperative game if and only if $\frac{\pi_s - \pi_l}{r} = R_l$. If $\frac{\pi_s - \pi_l}{r} = R_l$, then the first-best equilibrium is the unique equilibrium.
\end{proposition}
The proof technique is similar to the proof of Theorem \ref{thm:noncoop_efficient}. The difference here is that instead of working with the HJB equation and the public belief as a state variable, I instead characterize the time-dependent optimal control solution to the best-response problem. First, if the condition is violated, the same deviations from the observable case are still profitable without observability (either some agent has an incentive to lower effort earlier or stop effort later). If the condition holds, then a verification argument checks that the cutoff at time $t_{FB}$ is a best-response. To show uniqueness, I show that full effort is a strict best-response prior to time $t_{FB}$ regardless of whatever the other agents are doing and that stopping at $t_{FB}$ is the only best response if the other agents exert full effort up to time $t_{FB}$.

Given Proposition \ref{prop:unobservable} and Theorem \ref{thm:guarantee}, it then follows that the same sharing contract conditioning on winner/loser uniquely implements the efficient solution even when actions are unobservable by everyone: 
\begin{corollary}\label{corr:unobservable_guarantee}
    Even when actions are unobservable, the first-best solution is the unique equilibrium under sharing contract $c^W_{\alpha_I,\alpha_C}$ if and only if $G\left(c^W_{\alpha_I,\alpha_C} \right) = \pi_s$. Further, if $G\left(c^W_{\alpha_I,\alpha_C} \right) = \pi_s$, the first-best solution is the unique outcome.
\end{corollary}

\subsubsection{Examples}
I now present two examples to illustrate the insights from Theorem \ref{thm:guarantee}.

\paragraph{Drug Discovery}
Consider a regulator overseeing pharmaceutical research in a competitive industry. That is, effort here is analogous to investment in research. Suppose that in the absence of regulation, the winner receives only a flow continuation reward $\pi_w$ (the future profit from sales of the new drug relative to the current drug). The losing firms experience a \textit{reduced} flow continuation reward $\pi_l < \pi_s$ from imperfect competition against a superior drug, and there are no instantaneous rewards $(R = 0)$. The firms are secretive about their research, and hence their research investment choices are unobservable to their competitors and to the regulator. Hence, the regulator can only contract on the observation of winner/losers (for example, by approving a successful drug). The regulator in particular does not know anything about the timing of the game (i.e., when the research started) or the actions of the players. 

Theorem \ref{thm:guarantee} and Corollary \ref{corr:unobservable_guarantee} suggest that there is a simple way the regulator can restore efficiency by only enforcing ex post transfers that occur after a breakthrough. Specifically, the regulator can optimally offer ``participation trophies,'' or payments to the losers, to restore efficiency. In particular, there is a unique transfer $T_l$ paid to each loser that restores efficiency. In particular, $T_l$ is pinned down by
\begin{equation} \label{eqn:efficient_transfers}
T_l = \frac{\pi_s - \pi_l}{r} .
\end{equation}
If the regulator can commit to enforcing this transfer being paid to each loser, the regulator can restore efficiency by charging the winner the transfer that restores budget balance:
$ T_w = -(N-1)T_l $.

In particular, the regulation scheme above does not require the regulator to know anything about when the research started, how long the research has been going on, or anything about the history of research investment. Thus, in a scenario where it is at all costly or impossible for the regulator to monitor and contract on the actions of the firms, the regulator can still achieve the first-best.

\paragraph{Natural Resource Exploitation} 

Consider multiple agents involved in extracting a resource from a common area; for the sake of this example, imagine fishermen fishing in a lake. Fishing comes at an opportunity cost $\pi_s$, and the presence of fish is initially unknown. Fishing catches fish at a flow rate $\lambda$ if there are fish in the lake, and fishermen value each fish at $R$. The first fish that is caught corresponds to the breakthrough, so $R_w = R$ and $R_l = 0$. The continuation payoff to each fisherman after the first fish is caught is $\pi_w = \pi_l = \lambda R$, as they now know there are fish in the lake.

In this scenario, the efficiency-restoring transfers as characterized by \eqref{eqn:efficient_transfers} are negative for the losers, since $\pi_l = \lambda R > \pi_s$. That is, the efficient transfer scheme that conditions only on winning/losing requires the losers to pay a small amount to the winner to compensate for the net positive informational externality provided by the winner. In practice (and perhaps befitting a recreational setting), such a transfer scheme could be implemented by a betting pool for the first fisherman to catch a fish; that is, the fishermen each initially contribute $\frac{\pi_s - \pi_l}{r}$ into a pot, and the first person to catch a fish wins the pot (with the pot being redistributed back to each agent if everyone goes home after not catching fish). Such a scheme provides enough incentive to ensure that all fishermen continue fishing, even when in the original noncooperative setting they would rather go home.

\subsection{Contracting on Effort}
In some research environments, it might be infeasible, costly, or even taboo to contract on winning/losing. For example, consider two mathematicians collaborating on proving a theorem in the hopes of claiming a monetary prize for doing so. Provided that they are collaborating, it may seem underhanded for the mathematician who figures out the main proof idea to take sole credit and take the prize. 

In the baseline model of the experimentation game, the breakthroughs were attributable, meaning that breakthrough occurred for a publicly identifiable agent that I called the winner. In other research environments, it may be the case that breakthroughs are nonattributable, intentional or not, that is, a breakthrough occurs but it cannot be directly assigned to a single agent. A naive first approach might be that all participants share the experimentation rewards fairly; however, the following result (a consequence of Theorem \ref{thm:noncoop_efficient}) shows that this is inefficient. 

\begin{corollary}\label{corr:equal_sharing}
    Suppose that all agents agree to split the rewards evenly in the event of a breakthrough: that is, the contract takes the form
    \[ c(\cdot) = \left(\left\{\frac{R}{N}\right\}_i, \left\{\frac{\Pi}{N}\right\}_i \right). \]
 Then, any MPE of the respective game is inefficient. 
\end{corollary}
\begin{proof}
    Note that an even split of everything implies that this game is equivalent to a specific instance of the baseline experimentation game where $R_w = R_l = R/N$, and $\pi_w = \pi_l = \Pi/N > \pi_s$. Hence, $\pi_l > \pi_s$, so $\pi_l + r R_l > \pi_s$ and the corollary follows from Theorem \ref{thm:noncoop_efficient}.
\end{proof}

The corollary shows that an ex ante fair split of the rewards cannot result in an efficient outcome, so an efficient contract must still condition on some part of the outcome of the experimentation game. Recall that the previous subsection showed that conditioning the contract on the observation of winner/loser was sufficient to restore efficiency. An important insight from the baseline game is that conditional on a breakthrough occurring in an infinitesimal time interval $[t, t+dt)$, the probability of agent $i$ being the winner is $\frac{k_i(t)}{K(t)}$, where $K$ is instantaneous total flow effort and $k_i$ is the instantaneous flow effort of agent $i$. Hence, the intuitive extension to restore efficiency when outcomes are nonattributable is to condition the sharing contract on $k_i(\tau) / K(\tau)$, the instantaneous share of total flow effort at time of breakthrough.\footnote{One could condition contracts on much stronger instruments, such as the full history of effort. However, I show an efficiency result, and hence I derive contracts to restore efficiency using as weak of a contract as possible.} Using this, we can define an analogue of the sharing contracts. 

Formally, let $k(\tau) = \{ k_i (\tau) \}_i$ denote the vector endpoint of the effort profile path, and let $K(\tau) = \sum_i k_i(\tau)$. Consider the family of contracts that splits an $\alpha_I$ share of the instantaneous reward and $\alpha_C$ share of the continuation reward based on the effort profile at time of breakthrough:
\[ c^K_{\alpha_I, \alpha_C}(k(\tau)) =  \left(\left\{ \left(\alpha_I \frac{k_i(\tau)}{K(\tau)} + \frac{1-\alpha_I}{N}\right) \cdot R\right\}_i, \left\{ \left(\alpha_C \frac{k_i(\tau)}{K(\tau)} + \frac{1-\alpha_C}{N}\right) \cdot \Pi \right\}_i\right). \]
I define a guarantee of the sharing contract analogously to \eqref{eqn:def_guarantee} (and slightly abuse the notation $G$):
\begin{equation}\label{eqn:def_guarantee2}
    G(c^K_{\alpha_I, \alpha_C}) = rR\left(\frac{1-\alpha_I}{N} \right) + \Pi\left(\frac{1-\alpha_C}{N} \right).
\end{equation} 

The following result is the analogue of Theorem \ref{thm:guarantee} for nonattributable breakthroughs. 

\begin{theorem}\label{thm:guarantee_2}
    The first-best solution is an MPE of the game under sharing contract $c^K_{\alpha_I, \alpha_C}$ if and only if $G\left(c^K_{\alpha_I, \alpha_C}\right)  = \pi_s$. Furthermore, if $G\left(c^K_{\alpha_I, \alpha_C}\right)  = \pi_s$, then the first-best solution is the unique MPE. 
\end{theorem}

The nontrivial observation is that in the nonattributable breakthrough model here, conditioning on the instantaneous share of total effort $(k_i(\tau) / K(\tau))$ has the same impact on incentives as conditioning on winning in the baseline experimentation model. The proof is provided in the Appendix.

Although much of the analysis seems similar to the previous subsection with contracts conditional on winner/losers, in this setting, the payoffs of the game on-path are ex post symmetric; that is, on the equilibrium path, the payoffs to all agents are split evenly ex post, which is not true for equilibria induced by contracts conditional on winner/loser. In the previous subsection, the winning agent obtained a different payoff than all the losing agents. Under contracts conditioning on effort however, since $k_i(\tau) = k_j(\tau)$ on-path, all agents earn the same payoffs ex post. Thus, another reason why a regulator or contest designer might want to condition payoffs on effort is to ensure that payoffs after the experimentation game are symmetric.

It is important to note that despite the fact that payoffs are symmetric ex post on-path, Corollary \ref{corr:equal_sharing} shows that the regulator or contest designer cannot promise symmetric payoffs ex ante. Under contracts conditioning on effort, since $\Pi/N > \pi_s$, by applying the efficiency condition, we have
\[ 0 > \pi_s - \frac{\Pi}{N} = G\left(c^K_{\alpha_I, \alpha_C}\right)  - \frac{\Pi}{N} =  r\left(1-\alpha_I\right)\frac{R}{N}  - \alpha_C \frac{\Pi}{N}.  \]
Since the first term $r\left(1-\alpha_I\right)\frac{R}{N} $ is nonnegative, this implies that $\alpha_C$ must be strictly positive. Hence, in any efficient sharing contract, it is \textit{necessary} to condition at least some of the flow continuation payoff on the effort share, even if all agents split this evenly on-path. 

\subsubsection{Examples}
Once again, I provide two examples to illustrate the insights of Theorem \ref{thm:guarantee_2}.

\paragraph{Academic Collaboration} Return to the mathematician scenario; that is, a group of mathematicians are working on proving a theorem, with a monetary prize. Given that it might be taboo amongst collaborators to claim credit for a discovery, none of the mathematicians can claim to be the winner and take the prize. The insights of Theorem \ref{thm:guarantee_2} suggest that agreeing ex ante to split the prize if a breakthrough occurs is not efficient and results in underexertion of effort; rather, the mathematicians should agree to split part of the reward based on their share of total effort when the breakthrough is discovered and split the rest of the reward fairly.

\paragraph{Delegated Research} Consider a group of computer manufacturing firms that have decided to invest in a semiconductor startup that may or may not produce a successful chip. Here, the firms are not conducting the research themselves, but have delegated the research to the startup. In particular, it makes no sense to have winners/losers here, since a breakthrough for the startup is a breakthrough for all the investing firms. Abstracting from agency concerns, suppose that each unit of investment increases the flow rate of a chip breakthrough by $\lambda$ and that a breakthrough would increase total industry profits to $\Pi$. Investing in the startup is costly, coming at an opportunity cost $\pi_s$ to the computer manufacturers. As before, the computer manufacturers observe each others' investment decisions. The startup, on the other hand, is focused on research and finds it costly to pay attention to the investment decisions of the firms. Theorem \ref{thm:guarantee_2} shows that the startup can agree to only examine the investment profile at the time of breakthrough and award the rights to its new chip based on the investment profile at that time. By doing so, the noncooperative equilibrium is self-policing; the computer manufacturing firms invest the first-best amounts of their own volition.

\paragraph{Remark on the guarantee:} Theorems \ref{thm:guarantee} and \ref{thm:guarantee_2} have shown that when the guarantee exactly equals $\pi_s$, the contract restored efficiency. However, the guarantee naturally corresponds to the loser's value in the general game considered in Section \ref{sec:noncoop}. Intuitively, an inefficient contract will induce overexperimentation or underexperimentation if the guarantee is too small or too large, respectively (see Appendix \ref*{sec:inefficient} for formal results). 
\section{Extension: Resource Heterogeneity}
\label{sec:extensions}

In the baseline model, all agents have a unit measure of a resource that can be allocated between the research project and the status quo. The main insights extend to the case where agents have heterogeneous resources to invest between the research project and the status quo. In this section, suppose that instead, agent $i$ has a total measure $\mu_i$ of effort resources to be allocated between the status quo and research project. Thus, the action choice of agent $i$ is selecting $k_i(t) \in [0, \mu_i]$ (as opposed to $[0,1]$). Thus, the flow payoff from the status quo for agent $i$ becomes $\pi_s(\mu_i - k_i(t))$, and the flow rate of breakthrough arrival is still $\lambda k_i(t)$. Let $M = \sum_i \mu_i$ be the total measure of resource available to all agents. Since there is heterogeneity among agents, it is natural to allow payoffs to be heterogeneous; that is, each agent $i$ has a $\pi_{w, i}$ and $\pi_{l, i}$ continuation payoff from winning and losing, respectively. Similarly, I define heterogeneous instantaneous lump-sum payoffs $R_{w,i}$ and $R_{l,i}$. For the cooperative problem to remain agnostic to which agent wins, fix the total instant and continuation payoffs after a breakthrough to a constant:
\[ \Pi = \pi_{w,i} + \sum_{j \neq i} \pi_{l, j} \]
and
\[ R = R_{w,i} + \sum_{j \neq i} R_{l, j} \]
for all $i$. 

Once again, I assume that a breakthrough is overall welfare-improving; so $M \pi_s < \Pi$. I also still assume that every agent benefits from winning; that is, the flow continuation reward from winning is $\pi_{w,i} > \mu_i \pi_s$ for all $i$. 

Since in this extension the reasoning follows the same process as in the previous sections of the paper, I relegate formal proofs to the Appendix and state the analogues of the main results I established in the baseline case. The analogue of the first-best solution in Theorem \ref{thm:cooperative} is as follows:

\begin{proposition}\label{prop:hetero_cooperative}
    In the first-best solution with resource heterogeneity, all agents exert full effort $k_i(p) = \mu_i$ for $p > p^H_{FB}$, where the threshold is defined by 
    \begin{equation}\label{eqn:het_pFB}
        p^H_{FB} = \frac{\pi_s}{\lambda R + \frac{\lambda}{r}(\Pi - M\pi_s)}.
    \end{equation}
    For $p < p^H_{FB}$, $k_i(p) = 0$. 
\end{proposition}

Note that the belief threshold in \eqref{eqn:het_pFB} compared to \eqref{eqn:belief_fb} has an $M \pi_s$ in the denominator rather than $N \pi_s$. The analogue of Theorem \ref{thm:noncoop_efficient} is then as follows:
\begin{proposition}\label{prop:hetero_noncoop}
    In the experimentation game with resource heterogeneity, the first-best solution is an MPE of the noncooperative game if and only if 
    \[ \frac{\mu_i \pi_s - \pi_{l,i}}{r} = R_{l,i} \]
    for all $i$. If the condition holds, then the first-best solution is the unique MPE. 
\end{proposition}
Note that one implication of Proposition \ref{prop:hetero_noncoop} is that the payoffs to each agent that loses need to be generically different; that is, heterogeneity in the resources available for research implies that for efficiency to be an MPE of the noncooperative game, agents need identity-dependent continuation values upon losing. However, there is a straightforward extension of the contracts discussed in Section \ref{sec:conclusion} that incorporates the necessary heterogeneity to restore efficiency. Consider sharing contracts conditional on winner/losers that now take the form
\[ c^{W,H}_{\alpha_I, \alpha_C}(w) =  \left(\left\{ \left(\alpha_I \mathbbm{1}[w = i] + (1-\alpha_I) \frac{\mu_i}{M}\right) \cdot R\right\}_i, \left\{ \left(\alpha_C \mathbbm{1}[w = i] + (1-\alpha_C)\frac{\mu_i}{M}\right) \cdot \Pi \right\}_i\right) \]
and sharing contracts conditional on effort that take the form
\[ c^{K,H}_{\alpha_I, \alpha_C}(k(\tau)) =  \left(\left\{ \left(\alpha_I \frac{k_i(\tau)}{K(\tau)} +  (1-\alpha_I) \frac{\mu_i}{M}\right) \cdot R\right\}_i, \left\{ \left(\alpha_C \frac{k_i(\tau)}{K(\tau)} +  (1-\alpha_C) \frac{\mu_i}{M}\right) \cdot \Pi \right\}_i\right). \]
Note that now, the key change is that the guaranteed reward has a fraction $\mu_i/M$ rather than $1/N$. The analogous definition of a guarantee (for both types of contracts) then gives each agent a continuation flow reward that is proportional to $\mu_i$, the effort resource available to the agent:
\[ G_i\left(c^{\cdot, H}_{\alpha_I, \alpha_C} \right) = rR\left(1-\alpha_I\right)\frac{\mu_i}{M} + \Pi\left(1-\alpha_C\right)\frac{\mu_i}{M} . \]
It will be useful to define the normalized guarantee per unit resource as independent of $i$:
\[ g\left(c^{\cdot, H}_{\alpha_I, \alpha_C} \right) = \frac{G_i\left(c^{\cdot, H}_{\alpha_I, \alpha_C} \right)}{\mu_i} = rR\left(1-\alpha_I\right)\frac{1}{M} + \Pi\left(1-\alpha_C\right)\frac{1}{M}.  \]
\begin{proposition}\label{prop:hetero_contracts}
    In the experimentation game with resource heterogeneity, the first-best solution is an MPE of the game under sharing contract $c^{\cdot, H}_{\alpha_I, \alpha_C}$ if and only if $g\left(c^{\cdot, H}_{\alpha_I, \alpha_C}\right)  = \pi_s$. Furthermore, if $g\left(c^{\cdot, H}_{\alpha_I, \alpha_C}\right)  = \pi_s$, then the first-best solution is the unique MPE. 
\end{proposition}

Note that in the extension, the efficiency condition requires that the normalized guarantee $g$ equals $\pi_s$; this implies that the individual agent guarantee $G_i$ must be equal to $\mu_i \pi_s$, or that the individual agents are guaranteed an amount proportional to the measure of effort resource they have available.

\section{Conclusion}
\label{sec:conclusion}
In conclusion, this paper shows that very little information is necessary to remedy inefficiencies in strategic experimentation. In particular, it is sufficient for sharing contracts to condition on winning/losing or effort at time of breakthrough (but it is not necessary to condition on both or more). 

While the formal analysis was constrained to a specific model, this theoretical work offers important insights for thinking about research. First, the condition for efficiency when there are breakthrough payoff externalities is that breakthroughs must have a neutral impact on the losers. As much of the contest literature has focused on thinking about how to award winners, the analysis in this paper suggests that the key to understanding whether the amount of research conducted in such an environment is socially efficient is to consider how the \textit{losers} weigh the arrival of the breakthrough against the status quo. Second, the existence of simple contracts that restore efficiency suggests a method for sharing rewards for joint projects. The main insight is that the guarantee (or what agents are promised independent of their effort choices) must match their status quo opportunity cost of research effort. Indeed, these sharing contracts restore efficiency in a self-enforcing way; provided a contract that awards winners and losers in the right way, it becomes unnecessary to observe or contract on the actions of the other agents. On the other hand, if it is impractical or infeasible to identify the winner/losers, it is also sufficient for contracts to condition on effort shares at the time of breakthrough.

Finally, this model is a step in extending strategic experimentation models toward capturing the reality of research. Future work could incorporate other features of research that this model does not address. For example, collaboration sometimes arises because of complementarities between agents' effort. While this paper demonstrates that collaborating alleviates inefficiency in the absence of complementarities, a potential direction for further study is the interplay among technological complementarities, informational externalities, and payoff externalities. Another feature of research that this model does not capture is technological dependence on historical effort. Future work could extend the breakthrough-generating technology to also account for factors such as human capital accumulation.

\newpage

\bibliographystyle{agsm}
\bibliography{biblio}

@article{krc2005,
  title={Strategic experimentation with exponential bandits},
  author={Keller, Godfrey and Rady, Sven and Cripps, Martin},
  journal={Econometrica},
  volume={73},
  number={1},
  pages={39--68},
  year={2005},
  publisher={Wiley Online Library}
}

@article{kr2011,
  title={Negatively correlated bandits},
  author={Klein, Nicolas and Rady, Sven},
  journal={The Review of Economic Studies},
  volume={78},
  number={2},
  pages={693--732},
  year={2011},
  publisher={Oxford University Press}
}

@article{hkr2022,
  title={Overcoming free-riding in bandit games},
  author={H{\"o}rner, Johannes and Klein, Nicolas and Rady, Sven},
  journal={The Review of Economic Studies},
  volume={89},
  number={4},
  pages={1948--1992},
  year={2022},
  publisher={Oxford University Press}
}

@article{bh1999,
  title={Strategic experimentation},
  author={Bolton, Patrick and Harris, Christopher},
  journal={Econometrica},
  volume={67},
  number={2},
  pages={349--374},
  year={1999},
  publisher={Wiley Online Library}
}

@article{bh2011,
Author = {Bonatti, Alessandro and Hörner, Johannes},
Title = {Collaborating},
Journal = {American Economic Review},
Volume = {101},
Number = {2},
Year = {2011},
Month = {April},
Pages = {632-63},
DOI = {10.1257/aer.101.2.632},
URL = {https://www.aeaweb.org/articles?id=10.1257/aer.101.2.632}}

@article{al2016,
  title={The role of information in innovation and competition},
  author={Akcigit, Ufuk and Liu, Qingmin},
  journal={Journal of the European Economic Association},
  volume={14},
  number={4},
  pages={828--870},
  year={2016},
  publisher={Oxford University Press}
}

@article{kr2010,
  title={Strategic experimentation with Poisson bandits},
  author={Keller, Godfrey and Rady, Sven},
  journal={Theoretical Economics},
  volume={5},
  number={2},
  pages={275--311},
  year={2010},
  publisher={Wiley Online Library}
}

@article{rsv2007,
  title={Social learning in one-arm bandit problems},
  author={Rosenberg, Dinah and Solan, Eilon and Vieille, Nicolas},
  journal={Econometrica},
  volume={75},
  number={6},
  pages={1591--1611},
  year={2007},
  publisher={Wiley Online Library}
}

@article{rsv2013,
  title={On games of strategic experimentation},
  author={Rosenberg, Dinah and Salomon, Antoine and Vieille, Nicolas},
  journal={Games and Economic Behavior},
  volume={82},
  pages={31--51},
  year={2013},
  publisher={Elsevier}
}

@article{kr2015,
  title={Breakdowns},
  author={Keller, Godfrey and Rady, Sven},
  journal={Theoretical Economics},
  volume={10},
  number={1},
  pages={175--202},
  year={2015},
  publisher={Wiley Online Library}
}

@article{dks2020,
  title={Strategic experimentation with asymmetric players},
  author={Das, Kaustav and Klein, Nicolas and Schmid, Katharina},
  journal={Economic Theory},
  volume={69},
  number={4},
  pages={1147--1175},
  year={2020},
  publisher={Springer}
}

@article{hkl2017,
  title={Contests for experimentation},
  author={Halac, Marina and Kartik, Navin and Liu, Qingmin},
  journal={Journal of Political Economy},
  volume={125},
  number={5},
  pages={1523--1569},
  year={2017},
  publisher={University of Chicago Press Chicago, IL}
}

@book{bd1997,
  title={Optimal control and viscosity solutions of Hamilton-Jacobi-Bellman equations},
  author={Bardi, Martino and Dolcetta, Italo Capuzzo and others},
  volume={12},
  year={1997},
  publisher={Springer}
}

@article{d2018,
  title={Strategic experimentation with asymmetric information},
  author={Dong, Miaomiao},
  journal={Unpublished Paper, Pennsylvania State University.[1027]},
  year={2018}
}

@article{t2021,
Author = {Thomas, Caroline D.},
Title = {Strategic Experimentation with Congestion},
Journal = {American Economic Journal: Microeconomics},
Volume = {13},
Number = {1},
Year = {2021},
Month = {February},
Pages = {1-82},
DOI = {10.1257/mic.20170187},
URL = {https://www.aeaweb.org/articles?id=10.1257/mic.20170187}}

@article{m2004,
  title={Mechanism design with interdependent valuations: Efficiency},
  author={Mezzetti, Claudio},
  journal={Econometrica},
  volume={72},
  number={5},
  pages={1617--1626},
  year={2004},
  publisher={Wiley Online Library}
}

@article{as2013,
 ISSN = {00129682, 14680262},
 URL = {http://www.jstor.org/stable/23524323},
 author = {Susan Athey and Ilya Segal},
 journal = {Econometrica},
 number = {6},
 pages = {2463--2485},
 publisher = {[Wiley, Econometric Society]},
 title = {AN EFFICIENT DYNAMIC MECHANISM},
 urldate = {2023-05-23},
 volume = {81},
 year = {2013}
}

@article{bv2002,
author = {Bergemann, Dirk and Välimäki, Juuso},
title = {Information Acquisition and Efficient Mechanism Design},
journal = {Econometrica},
volume = {70},
number = {3},
pages = {1007-1033},
doi = {https://doi.org/10.1111/1468-0262.00317},
url = {https://onlinelibrary.wiley.com/doi/abs/10.1111/1468-0262.00317},
eprint = {https://onlinelibrary.wiley.com/doi/pdf/10.1111/1468-0262.00317},
year = {2002}
}

@article{jm2003,
author = {Jehiel, Philippe and Moldovanu, Benny},
title = {Efficient Design with Interdependent Valuations},
journal = {Econometrica},
volume = {69},
number = {5},
pages = {1237-1259},
doi = {https://doi.org/10.1111/1468-0262.00240},
url = {https://onlinelibrary.wiley.com/doi/abs/10.1111/1468-0262.00240},
eprint = {https://onlinelibrary.wiley.com/doi/pdf/10.1111/1468-0262.00240},
year = {2001}
}

@article{bv2010,
 ISSN = {00129682, 14680262},
 URL = {http://www.jstor.org/stable/40664492},
 author = {Dirk Bergemann and Juuso Välimäki},
 journal = {Econometrica},
 number = {2},
 pages = {771--789},
 publisher = {[Wiley, Econometric Society]},
 title = {THE DYNAMIC PIVOT MECHANISM},
 urldate = {2023-06-08},
 volume = {78},
 year = {2010}
}

\newpage

\appendix

\section{Omitted Proofs}
\subsection*{Derivation of the Cooperative HJB Equation \eqrefb{eqn:hjb_cooperative}}
Define the belief $p$ as the state of the system, and $K$ as the control variable of the cooperative planner; equation \eqrefb{eqn:belief_law} generates a control constraint:
\[ \dot{p} = - K \lambda p(1-p) \]
The first-best solution maximizes the total payoff of all agents:
\[ \int^\tau re^{-rt}\pi_s(N - K) \ dt + re^{-r\tau}R + e^{-r\tau}\Pi \]
Since the total payoff only depends on the total effort $K$ and not the individual $k_i$, the problem takes the form of an optimal control problem with a one-dimensional control variable and state. We follow standard methods to derive the Hamilton-Jacobi-Bellman equation. Let $u(p)$ denote the maximized value function under belief $p$, or 
\[ u(p) = \max_{K} \mathbb{E}\left[ \left. \int^\tau re^{-rt}\pi_s(N - K) \ dt + re^{-r\tau}R + e^{-r\tau}\Pi \right\vert p \right] \]
Note that the distribution of $\tau$ depends on $p$. Consider an infinitesimal time increment $\Delta$, and suppose we fix a constant policy $K$ over the time increment $\Delta$. The payoff from setting $K$ over this increment is
\begin{align*}
    &\mathbb{P}\left[\tau \in [0, \Delta)\right] \mathbb{E}\left[ \left. \int^\tau re^{-rt}\pi_s(N - K) \ dt + re^{-r\tau}R + e^{-r\tau}\Pi\right\vert \tau \in [0,\Delta)  \right]  \\
    + &\mathbb{P}\left[\tau \not \in [0, \Delta)\right]\left( \left(1 - e^{-r \Delta}\right)\pi_s(N - K) + e^{-r \Delta} u(p + dp)\right)
\end{align*}  
where $p + dp$ is the evolution of the state belief according to the control constraint (belief law of motion).
Given that this payoff results from setting policy $K$ over the increment, the probability $\tau$ is in $[0, \Delta)$ is  $pK\lambda \Delta$. So the payoff from setting $K$ over this increment (and then optimizing the choice of $K$) is given by 
\begin{align*}
    M_\Delta(p) = \max_K & \begin{Bmatrix}  pK \lambda \Delta \mathbb{E}\left[ \left.\int^\tau re^{-rt}\pi_s(N - K) \ dt + re^{-r\tau}R + e^{-r\tau}\Pi \right\vert \tau \in [0,\Delta)  \right]  \\
    \left.+(1 - pK\lambda \Delta)\left( \left(1 - e^{-r \Delta}\right)\pi_s(N - K) + e^{-r \Delta} u(p + dp)\right) \right\} \end{Bmatrix}
\end{align*} 
As we take the increment $\Delta \to 0$, the value of the above expression should approach the optimal value of the problem $u(p)$. By variational calculus, since $u$ was the functional optimum, the variational derivative $\frac{1}{\Delta}(u - M_\Delta) \to 0$ as $\Delta \to 0$. Consider exactly the variational difference
\begin{align*}
    \frac{1}{\Delta}\left( u(p) -\max_K \left\{ pK \lambda \Delta E +(1 - pK\lambda \Delta)\left( \left(1 - e^{-r \Delta}\right)\pi_s(N - K) + e^{-r \Delta} u(p + dp)\right) \right\}  \right)
\end{align*} 
where 
\begin{equation*}
   E(\Delta)  =  \mathbb{E}\left[ \left.\int^\tau re^{-rt}\pi_s(N - K) \ dt + re^{-r\tau}R + e^{-r\tau}\Pi \right\vert \tau \in [0,\Delta)  \right] 
\end{equation*}
As argued, as $\Delta \to 0$ this expression should converge to 0: 
\begin{align*}
    \lim_{\Delta \to 0} \frac{1}{\Delta}\left(u(p) - \max_K \left\{ pK \lambda \Delta E(\Delta) +(1 - pK\lambda \Delta)\left( \left(1 - e^{-r \Delta}\right)\pi_s(N - K) + e^{-r \Delta} u(p + dp)\right) \right\} \right) = 0
\end{align*}  
Pulling out a $u(p + dp)$, and distributing the $1/\Delta$ term,
\begin{align*}
    \lim_{\Delta \to 0} \max_K &  \left\{ \left( \frac{u(p) - u(p + dp)}{\Delta} \right) \right.- \frac{1}{\Delta}pK \lambda \Delta E(\Delta) \\
    &\left. +\frac{1}{\Delta}(1 - pK\lambda \Delta)\left( \left(1 - e^{-r \Delta}\right)\pi_s(N - K) +  e^{-r \Delta} u(p + dp)\right) - \frac{1}{\Delta}u(p + dp) \right\} = 0
\end{align*}  
Rearranging terms, 
\begin{align*}
    \lim_{\Delta \to 0} \max_K   \begin{Bmatrix} \left( \frac{u(p) - u(p + dp)}{\Delta} \right)  -  pK \lambda E(\Delta) \\
     +(1 - pK\lambda \Delta)\left(\frac{1 - e^{-r \Delta}}{\Delta}\pi_s(N - K) - \frac{1- e^{-r \Delta}}{\Delta} u(p + dp)\right) - pK\lambda u(p + dp) \end{Bmatrix} = 0
\end{align*}  
Note that $\lim_{\Delta \to 0} \frac{1 - e^{-r\Delta}}{\Delta} = r$, $\tau\mid_{\tau < \Delta} \to 0$ as $\Delta \to 0$, and from the belief law of motion, $u(p) - u(p + dp) \to K\lambda p (1-p) u'(p) \Delta $ uniformly in $K$. Thus, the maximizing expression converges uniformly in $K$ as $\Delta \to 0$, and so we can interchange the maximization and the limit. Using this, interchange the max and the limit and the expression simplifies considerably, as when $\Delta \to 0$, $E(\Delta) \to 0$, and $1 - pK\lambda \Delta \to 1$: 
\begin{align*}
    \max_K  &\left\{ - u'(p) K\lambda p(1-p) +pK \lambda \left( rR + \Pi  \right)+  \left( r\pi_s(N - K) - ru(p)\right)  - pK\lambda u(p) \right\}=0
\end{align*}  
\begin{align*}
    r u(p) &= \max_K  \left\{r\pi_s(N - K)  + pK\lambda \left(r R + \Pi  - u(p) - u'(p) (1-p)\right) \right\} \\
    u(p) &= \max_K  \left\{ \pi_s(N - K) + pK\lambda R + pK\frac{\lambda}{r} \left( \Pi  - u(p) - u'(p) (1-p)\right) \right\} 
\end{align*}  
Dividing through by $N$ to get the average Bellman value of an agent in the first-best solution:
\begin{align*}
    \frac{u(p)}{N} &= \left(\pi_s + \max_K \left[K\left(p\frac{\lambda}{r}\left(\frac{\Pi}{N} - \frac{u(p)}{N} - (1-p)\frac{u'(p)}{N} \right) - \frac{c(p)}{N} \right) \right] \right)  \\
    V_N(p) &= \pi_s + \max_K \left[K\left(p\frac{\lambda}{r}\left(\frac{\Pi}{N} - V_N(p) - (1-p)V_N'(p) \right) - \frac{c(p)}{N} \right) \right] 
\end{align*}
\subsection*{Proof of Theorem \ref*{thm:cooperative}}
I solve the HJB in Equation \eqrefb{eqn:hjb_cooperative}. It is immediate from the HJB that total effort $K$ depends on whether 
\[ p\frac{\lambda}{r}\left(\frac{\Pi}{N} - V_N(p) - (1-p)V_N'(p)  \right) \]
is larger than or smaller than $c(p)/N$. If it is larger, then optimally $K = N$, and if it is smaller, $K = 0$. When $K = N$, I obtain the following differential equation for $V_N$:
\[  V_N(p) = \pi_s + N\left(p\frac{\lambda}{r}\left(\frac{\Pi}{N} - V_N(p) - (1-p)V_N'(p) \right) - \frac{c(p)}{N} \right)   \]
\[  \left(1 + \frac{Np \lambda}{r}\right)V_N(p) + \frac{Np(1-p)\lambda}{r}V_N'(p) = \pi_s + \left(Np\frac{\lambda}{r}\left(\frac{\Pi}{N}\right) - \left(\pi_s - p\lambda R\right) \right)   \]
\begin{equation}\label{eqn:coop_effort_diffeq}
\left(1 + \frac{Np \lambda}{r}\right)V_N(p) + \frac{Np(1-p)\lambda}{r}V_N'(p) = p\lambda \left( \frac{\Pi}{r} + R \right)   
\end{equation}
I explicitly solve this differential equation. The definition of odds ratio is
\begin{equation*}
\Omega(p) := \frac{1-p}{p}  
\end{equation*}
Consider the function:
\[ \phi(p) = (1-p)\Omega(p)^{\frac{r}{N\lambda}} \]
Its derivative is 
\begin{align*}
\phi'(p) &= - \Omega(p)^{\frac{r}{N\lambda}} + \frac{r}{N\lambda}(1-p)\Omega(p)^{\frac{r}{N\lambda} - 1}\left(\frac{-1}{p^2} \right) \\
&= -\Omega(p)^{\frac{r}{N\lambda}}\left(1 + \frac{r}{N\lambda p}\right)
\end{align*}
Note that 
\begin{align*}
 &s\left(1 + \frac{Np \lambda}{r}\right)\phi(p) + \frac{Np(1-p)\lambda}{r}\phi'(p) \\
 &= \left(1 + \frac{Np \lambda}{r}\right)(1-p)\Omega(p)^{\frac{r}{N\lambda}} - \frac{Np(1-p)\lambda}{r}\Omega(p)^{\frac{r}{N\lambda}}\left(1 + \frac{r}{N\lambda p}\right) \\
 &=  \left(1 + \frac{Np \lambda}{r}\right)(1-p)\Omega(p)^{\frac{r}{N\lambda}} - (1-p)\Omega(p)^{\frac{r}{N\lambda}}\left(1 + \frac{N\lambda p}{r}\right) = 0
\end{align*}
So the solutions to the differential equation \eqref{eqn:coop_effort_diffeq} are 
\begin{equation}\label{eqn:coop_value_function_effort}
V_N(p) = p \frac{\lambda \left( \frac{\Pi}{r} + R \right)}{1 + \frac{N\lambda}{r}} + C\phi(p) 
\end{equation}
for some constant $C$. This characterizes the behavior of the value function when full effort occurs. Note that the first term denotes the expected payoff from committing to research for the rest of time, and the second term is the option value of being able to abandon research. Hence $C$ is nonnegative, and so the characterization in \eqref{eqn:coop_value_function_effort} is convex. 

To finish characterizing the solution, when $K = 0$, $V_N(p) = \pi_s$, and $V_N'(p) = 0$, so the smooth pasting and value-matching conditions pin down $C$ and the transition belief $p_{FB}$. Because the value function characterization from \eqref{eqn:coop_value_function_effort} is convex, there can exactly one point where these conditions can be satisfied. 

At the threshold belief $p_{FB}$ then, smooth pasting and value matching implies that the differential equation \eqref{eqn:coop_effort_diffeq} satisfies:
\[ \left(1 + \frac{Np \lambda}{r}\right)\pi_s = p\lambda \left( \frac{\Pi}{r} + R \right)   \]
\[ \pi_s = p\lambda \left( \frac{N}{r}\left(\frac{\Pi}{N} - \pi_s\right) + R \right)  \]
\[ p_{FB} = \frac{\pi_s}{\lambda \left( \frac{N}{r}\left(\frac{\Pi}{N} - \pi_s\right) + R \right)}  \]
A quick rearrangement gives the expression in \eqrefb{eqn:belief_fb}. Lastly, I solve for $C$, from the value matching condition:
\[ \pi_s = p_{FB} \frac{\lambda \left( \frac{\Pi}{r} + R \right)}{1 + \frac{N\lambda}{r}} + C\phi(p_{FB})  \]
\[ C = \frac{\pi_s - p_{FB} \frac{\lambda \left( \frac{\Pi}{r} + R \right)}{1 + \frac{N\lambda}{r}}}{\phi(p_{FB})} =  \frac{\pi_s\left(1 + \frac{N\lambda}{r}\right) - p_{FB} \lambda \left( \frac{\Pi}{r} + R \right)}{\left(1 + \frac{N\lambda}{r}\right) \phi(p_{FB})}=  \frac{\pi_s(1- p_{FB})\frac{N\lambda}{r}}{\left(1 + \frac{N\lambda}{r}\right) \phi(p_{FB})}\]
where the last step used the equation pinning down $p_{FB}$. 
So the value function is 
\begin{equation}\label{eqn:value_function_fb}
V_{FB}(p) = \begin{cases}
 \pi_s & p < p_{FB} \\
 p \frac{\lambda \left( \frac{\Pi}{r} + R \right)}{1 + \frac{N\lambda}{r}} + \frac{\pi_s\left(1 + \frac{N\lambda}{r}\right) - p_{FB} \lambda \left( \frac{\Pi}{r} + R \right)}{\left(1 + \frac{N\lambda}{r}\right) \phi(p_{FB})}\phi(p) & p \ge p_{FB}
\end{cases}
\end{equation}
Optimality follows from standard verification arguments; since the HJB admits a unique solution (if a continuous, differentiable solution exists), and we have constructed a continuous, differentiable solution, we are done.
\hfill \qedsymbol    
\subsection*{Derivation of the Baseline HJB Equation \eqrefb{eqn:hjb_baseline}}
The formal logic follows the same as in the derivation of the cooperative problem HJB; I provide the heuristic calculation here that parallels the formal argument (select a $dt$ and take a constant policy on the interval $[t, t + dt)$, setting the variational derivative to zero as the time increment goes to 0); heuristically, this gives
\begin{align*}
    u(p) &=&& \max_{k_i}&&\left[r \ dt \left( \left(1 - k_i \right)\pi_s + pk_i \lambda R_w + pK_{-i}(p)\lambda R_l \right) +  (1 - r \ dt)  p \ dt (\lambda k_i \pi_w +\lambda K_{-i}(p) \pi_l) \right.\\
    &&&&&  +\left. (1 - r \ dt)(1 - p (k_i + K_{-i}(p)) \lambda \ dt)(u(p) - (k_i + K_{-i}(p))\lambda p (1-p)u'(p) \ dt )  \right] \\
    u(p) &=&& \max_{k_i}&&\left[r \ dt \left[ \left(1 - k_i \right)\pi_s + p k_i \lambda R_w  + pK_{-i}(p)\lambda R_l \right]\right.+ u(p) - r \ dt  \ u(p)  \\
    &&&&& \left.+  p \ dt (\lambda k_i \pi_w +\lambda K_{-i}(p) \pi_l)  - (k_i + K_{-i}(p))\lambda p \ dt (u(p) + (1-p)u'(p))     \right] \\
    u(p) &=&& \max_{k_i}&&\left[ \left( \left(1 - k_i \right)\pi_s + p k_i \lambda R_w + p K_{-i}(p)\lambda R_l \right) \right. \\
    &&&&&\left.+ p \frac{\lambda}{r} \left(k_i \pi_w + K_{-i}(p) \pi_l  - (k_i + K_{-i}(p))(u(p) + (1-p)u'(p)) \right)  \right] 
\end{align*} 
Let $b_I$, $c_I$ be defined as in Equations \eqrefb{eqn:def_bI} and \eqrefb{eqn:def_cI}. Then the HJB becomes:
\begin{align*}
    u(p) &= \pi_s + K_{-i}(p)\left(p\lambda R_l + b_I(p,u,u') - p\frac{\lambda}{r}(\pi_w - \pi_l) \right) +  \max_{k_i} \left[k_i \left( b_I(p, u, u') - c_I(p) \right)  \right]
\end{align*}

\subsection*{Proof of Lemma \ref*{lem:viscosity}}
It suffices to check the regularity conditions necessary for Theorem 2.12 in \cite{bd1997} (A0 - A4). The action space is $[0,1]$, which is closed and compact (implying the first half of A0). The law of motion dictating how actions influence the state is the law of motion \eqrefb{eqn:belief_law}, which we can quickly confirm is continuous (implying A0), bounded (implying A1), and Lipschitz continuous in $p$ (implying A2 and A3). The discount rate $r > 0$, and payoffs are Lipschitz continuous in the action, implying A4. Thus, there exists a unique viscosity solution to the HJB \eqrefb{eqn:hjb_baseline}. 
\hfill \qedsymbol    

\subsection*{Proof of Lemma \ref*{lem:baseline_br}}
Recall that the HJB equation from \eqrefb{eqn:hjb_baseline} is:
\begin{align*} 
u(p) &&= \pi_s &+ K_{-i}(p)\left(p\lambda R_l + b_I(p,u,u')-p\frac{\lambda}{r}(\pi_w - \pi_l) \right) +  \max_{k_i} \left[k_i \left(b_I(p, u, u') - c_I(p) \right)  \right] \\
&&= \pi_s &+ K_{-i}(p) \left(b_I(p,u,u') - c_I(p) \right) + K_{-i}(p)\left( p\lambda R_l + c_I(p) -  \frac{p\lambda}{r}(\pi_w - \pi_l)\right) \\ 
&&& + \max_{k_i} \left[k_i \left( b_I(p, u, u') - c_I(p) \right)  \right].
\end{align*}
Consider $b_I(p, u, u') - c_I(p)$. In any equilibrium, if $k_i = 1$, this term must have been nonnegative. If $k_i = 0$, it must have been nonpositive. If $k_i \in (0,1)$, the term must have been zero. Equivalently, if $k_i = 1$ in equilibrium, it must be the case that
\[ u(p) \ge \pi_s + K_{-i}(p)\left(p\lambda R_l + c_I(p) -  \frac{p\lambda}{r}(\pi_w - \pi_l)\right).  \]
If $k_i = 0$, 
\[ u(p) \le \pi_s + K_{-i}(p)\left(p\lambda R_l + c_I(p) -  \frac{p\lambda}{r}(\pi_w - \pi_l)\right),  \]
and if $k_i \in (0,1)$
\[u(p) = \pi_s + K_{-i}(p)\left(p\lambda R_l + c_I(p) -  \frac{p\lambda}{r}(\pi_w - \pi_l)\right). \]
Combining the cases and substituting in for $c_I(p)$, one recovers the policy in \eqrefb{eqn:baseline_br}.
\hfill \qedsymbol    

\subsection*{Proof of Theorem \ref*{thm:noncoop_efficient}}
First, we show that the first-best solution is an MPE. Suppose all other agents are playing cutoff strategies at $p_{FB}$. We check that $V_{FB}$ defined in \eqref{eqn:value_function_fb} is a solution to the HJB. Then the HJB above $p_{FB}$ gives
\begin{align*} 
    u(p) = &\pi_s + (N-1) \left[p \lambda R_l + p\frac{\lambda}{r}(\pi_l - u(p) - (1-p)u'(p))\right] \\
    &+  \max_{k_i} \left[k_i \left( p\frac{\lambda}{r}(\pi_w - u(p) - (1-p)u'(p)) - (\pi_s - p \lambda R_w) \right)  \right] 
    \end{align*}
If $k_i = 1$, we get the same differential equation as the cooperative case, equation \eqref{eqn:coop_effort_diffeq}, and $V_{FB}$ by construction satisfies this differential equation for $p> p_{FB}$. Below $p_{FB}$ no other agents experiment, so the HJB implies that 
\begin{align*} 
    u(p) = &\pi_s +  \max_{k_i} \left[k_i \left( p\frac{\lambda}{r}(\pi_w - u(p) - (1-p)u'(p)) - (\pi_s - p \lambda R_w) \right)  \right] 
\end{align*}
Note that since $\pi_l + rR_l = \pi_s$, we have that for any $p$, 
\[ p\frac{\lambda}{r}(\pi_w - \pi_s) - (\pi_s - p\lambda R_w) = p\frac{\lambda}{r}(\Pi - N \pi_s) - (\pi_s - p\lambda R)\]
Then at $p < p_{FB}$, 
\begin{align*} 
    &\pi_s +  \max_{k_i} \left[k_i \left( p\frac{\lambda}{r}(\pi_w - V_{FB}(p) - (1-p)V'_{FB}(p)) - (\pi_s - p \lambda R_w) \right)  \right] \\
    =&\pi_s +  \max_{k_i} \left[k_i \left( p\frac{\lambda}{r}(\pi_w - \pi_s) - (\pi_s - p \lambda R_w) \right)  \right] \\
    =&\pi_s +  \max_{k_i} \left[k_i \left(  p\frac{\lambda}{r}(\Pi - N \pi_s) - (\pi_s - p\lambda R) \right)  \right] \\
    =&\pi_s +  \max_{k_i} \left[k_i \left(  p\frac{\pi_s}{p_{FB}} - \pi_s \right)  \right] = \pi_s = V_{FB}(p)
\end{align*}
Hence the HJB is satisfied below $p_{FB}$ as well. So $V_{FB}$ is a solution to the HJB, and therefore symmetric cutoff strategies at $p_{FB}$ is an MPE.

Now, I argue that if the condition fails, the efficient solution cannot be an MPE. 
If $\frac{\pi_s - \pi_l}{r} < R_l$, Lemma \ref*{lem:ptimes_pFB_comparison} implies that $p_\times > p_{FB}$, so the point $(p_{FB}, \pi_s)$ lies strictly interior in the half-plane below $\mathcal{D}_{N-1}$ (see Figure \ref*{fig:ptimes}); thus, playing a cutoff strategy at $p_{FB}$ cannot be a best-response by Lemma \ref*{lem:baseline_br}, since it would imply $k_i = 1$ below $\mathcal{D}_{N-1}$. For the other case, if $\frac{\pi_s - \pi_l}{r} > R_l$, Lemmas \ref*{lem:ptimes_pFB_comparison} and \ref*{lem:pI_characterization} imply that $p_I < p_{FB}$. Thus, the efficient solution cannot be an MPE, since it would imply that there exist an MPE where agents use weakly monotonic strategies but experimentation stops at $p_{FB} > p_I$, a contradiction of Lemma \ref*{lem:stopping_lower_bound}.

Lastly, I show that this is the only MPE. By Lemma \ref*{lem:stopping_lower_bound}, experimentation can never stop below $p_{FB}$, so all agents must stop exerting effort by $p_{FB}$. To show that this is the only MPE, I argue that in any equilibrium, each agent must be exerting full effort above $p_{FB}$. Suppose not, that some equilibrium exists where some agent $i$ has $k_i(p) < 1$ for $p > p_{FB}$. Let $u$ denote the value function of that agent. Since $rR_l + \pi_l = \pi_s$, any agent could obtain $\pi_s$ by playing $k_{-i} = 0$, and so $u(p) \ge \pi_s$. If $u(p) > \pi_s$, then $(p, u(p))$ lies above $\mathcal{D}_{K_{-i}}$ for any $K_{-i}$ so $k_i(p) < 1$ is a contradiction of Lemma 1. If $u(p) = \pi_s$, then by Lemma 1, the only way $k_i(p) < 1$ can be an equilibrium best-response policy requires $K_{-i}(p) = 0$, so the HJB implies that in order for $k_i < 1$ to be optimal, taking any sequence $p_n < p$, $p_n \to p$ and $u$ differentiable at $p_n$, we have  
\[ p_n\frac{\lambda}{r}(\pi_w - u(p_n) - (1-p_n)u'(p-n)) \le (\pi_s - p_n \lambda R_w)  \]
In the limit as $p_n \to p$, 
\[ -p\frac{\lambda}{r}(1-p)u'_-(p) \le \pi_s \left(1 - \frac{p}{p_I}\right)  \]
where $u'_-$ is the left-derivative (again, $u$ need not be differentiable). But this implies that the left-derivative of $u(p)$ is positive (since $p > p_{FB}=p_I$ the right-hand side is negative) and so there must exist some $p' < p$, such that $u(p') < \pi_s$, a contradiction of the fact that any agent can guarantee at least $\pi_s$ by always playing $k_i = 0$. Hence in either case of $u(p) = \pi_s$ or $u(p) > \pi_s$ it cannot be an equilibrium best-response to play $k_i(p) < 1$, and so the only equilibrium must be the first-best solution.
\hfill \qedsymbol    

\subsection*{Proof of Corollary \ref*{corr:unobservable_fb}}

    Note that the observability structure does not matter for the social planner, so the planner may as well have a public belief, allowing us to use Theorem \ref*{thm:cooperative}.

    The belief law of motion when all agents are exerting full effort on research gives 
    \[ \dot{p} = -N\lambda p(1-p) \]
    As before, let $\Omega$ denote the odds ratio. Then it is relatively straightforward to confirm that given the initial belief at $p(0)$, the belief path that satisfies the differential equation is 
    \[ p(t) = \frac{\exp(-N\lambda t)}{\Omega(p(0)) + \exp(-N\lambda t)} \]
    Setting the LHS to $p_{FB}$ and solving for $t$, we get 
    \[ p_{FB} = \frac{\exp(-N\lambda t_{FB})}{\Omega(p(0)) + \exp(-N\lambda t_{FB})} \]
    \[ 1 - p_{FB} = \frac{\Omega(p(0)}{\Omega(p(0)) + \exp(-N\lambda t_{FB})} \]
    \[ \Omega(p_{FB}) = \frac{\Omega(p(0)}{\exp(-N\lambda t_{FB})} \]
    \[ \exp(-N\lambda t_{FB})  = \frac{\Omega(p(0)}{\Omega(p_{FB})} \]
    \[ t_{FB} = \frac{1}{N\lambda}\ln\left( \frac{\Omega(p_{FB})}{\Omega(p(0)} \right) \] 
\hfill \qedsymbol

\subsection*{Proof of Proposition \ref*{prop:unobservable}}
    First, suppose $\frac{\pi_s - \pi_l}{r} \neq R_l$. From Theorem \ref*{thm:noncoop_efficient}, the efficient solution cannot be an equilibrium even when the actions are observable. Hence, when actions are unobservable, the same constructed deviations in the proof of Theorem \ref*{thm:noncoop_efficient} are still profitable, especially when these deviations are not detectable. 

    Now, suppose $\frac{\pi_s - \pi_l}{r} = R_l$. I first confirm that cutoff strategies at $t_{FB}$ are an equilibrium. Suppose all other agents are using a cutoff at $t_{FB}$. Consider the best-response problem of a single agent:

    \begin{gather*}
        \max_{k} \mathbb{E}\left[\int^\tau re^{-rt}\pi_s(1 - k(t)) \ dt + e^{-r\tau}\left( \frac{k(\tau)}{k(\tau) + K_{-i}(\tau)}(rR_w + \pi_w) + \frac{K_{-i}(\tau)}{k(\tau) + K_{-i}(\tau)}(rR_l + \pi_l)\right) \right]  \\
        \dot{p}(t) = - (k(t) + K_{-i}(t))\lambda p(t)(1-p(t)) \\
        K_{-i}(t) = \begin{cases} 
        N-1 & t < t_{FB} \\
        0 & t \ge t_{FB} 
        \end{cases}
    \end{gather*}
    Consider the subgame problem that occurs after $t_{FB}$. The agent's continuation problem is
    \begin{gather*}
        \max_{k} \mathbb{E}\left[\int^\tau re^{-rt}\pi_s(1 - k(t)) \ dt + e^{-r\tau}(rR_w + \pi_w) \right]  \\
        \dot{p}(t) = - k(t)\lambda p(t)(1-p(t)) 
    \end{gather*}
    This is exactly a single-agent (hence trivially cooperative) problem; from Theorem \ref*{thm:cooperative}, there exists a Markovian optimal policy, with cutoff at belief 
    \begin{align*}
        \frac{\pi_s/\lambda}{R_w + \frac{1}{r}(\pi_w - \pi_s)} &= \frac{\pi_s/\lambda}{R + \frac{1}{r}(\Pi - N\pi_s) - (N-1)\left(R_l + \frac{1}{r}(\pi_l - \pi_s) \right)} \\
        &= \frac{\pi_s/\lambda}{R + \frac{1}{r}(\Pi - N\pi_s) } = p_{FB}
    \end{align*} 
    Hence, if the belief $p(t_{FB})$ is at (or below) $p_{FB}$ following some agent strategy on time $[0, t_{FB})$, the agent must stop researching at $t_{FB}$. This also implies that if the belief is $p(t_{FB})$ at $t_{FB}$, then the agent's continuation value is given by
    $V_{FB}(p(t_{FB}))$, from equation \eqref{eqn:value_function_fb}.
    
    Now, consider the finite horizon $[0, t_{FB})$. Let $u(p, t)$ denote the value of the problem at time $t$ where the state is $p$. By Bellman's principle of optimality, we can similarly derive the HJB as before, where the value function is now dependent on both time and belief:
    \begin{align*} 
    u(p, t) =& \max_k \begin{bmatrix} r \dd t \left( \pi_s ( 1 - k) + p k \lambda R_w + p K_{-i}(t)\lambda R_l \right) + \lambda p \dd t\left( k \pi_w + K_{-i}(t) \pi_l \right)  \\ 
    + (1 - r \dd t) (1 - p(k + N-1)\lambda \dd t)(u(p + dp, t + dt)) \end{bmatrix} \\
    =&\max_k \begin{bmatrix} r \dd t \left( \pi_s ( 1 - k) + p k \lambda R_w + p K_{-i}(t)\lambda R_l \right) + \lambda p \dd t\left( k \pi_w + K_{-i}(t) \pi_l \right)  \\ 
    + (1 - r \dd t) (1 - p(k + N-1)\lambda \dd t)(u(p, t) + u_p(p, t)\dot{p} \dd t + u_t(p, t) \dd t) \end{bmatrix} \\
    =&\max_k \begin{bmatrix} r \dd t \left( \pi_s ( 1 - k) + p k \lambda R_w + p K_{-i}(t)\lambda R_l \right) + \lambda p \dd t\left( k \pi_w + K_{-i}(t) \pi_l \right)  \\ 
    + u(p, t) - r \dd t u(p, t) - p(k + K_{-i}(t))\lambda \dd t u(p,t) + u_p(p, t)\dot{p} \dd t + u_t(p, t) \dd t \end{bmatrix} \\
    u(p, t) =&\max_k \begin{bmatrix} \left( \pi_s ( 1 - k) + p k \lambda R_w + p K_{-i}(t)\lambda R_l \right) + \frac{\lambda}{r} p \left( k \pi_w + K_{-i}(t) \pi_l \right)  \\ 
    - p(k + K_{-i}(t))\frac{\lambda}{r}  u(p,t) + \frac{1}{r}\left( u_p(p, t)\dot{p} + u_t(p, t) \right) \end{bmatrix} 
    \end{align*}
    \begin{align}
    u(p,t) - \frac{1}{r}u_t(p,t) =& \pi_s + K_{-i}(t)p \lambda \left( R_l + \frac{1}{r}(\pi_l - u(p,t) - (1-p)u_p(p,t)) \right) \notag \\
    &+ \max_k \left[ k \left( - \pi_s + p \lambda R_w  + \frac{\lambda}{r} p \pi_w - p \frac{\lambda}{r}  u(p,t) - \frac{\lambda}{r} p(1-p)u_p(p, t)  \right) \right] \label{eqn:time_hjb}
    \end{align}
    Note that the belief, in absence of a breakthrough, is strictly decreasing, and the minimum the belief could be at $t_{FB}$ is exactly $p_{FB}$. Hence, the only domain of the HJB we have to consider is $t < t_{FB}$ and $p > p_{FB}$, where $K_{-i}(t) = N - 1$.

    The boundary condition for the terminal time is that $u(p, t_{FB}) = V_{FB}(p)$. Now, I verify that $u(p,t) = V_{FB}(p)$ satisfies the HJB together with the cutoff policy $k = 1$ when $p > p_{FB}$ and $t < t_{FB}$, $0$ otherwise. The solution by construction already satisfies the boundary condition. Plugging $V_{FB}$ into \eqref{eqn:time_hjb}, 
    \begin{align}
    V_{FB}(p) =& p \lambda \left( R + \frac{1}{r}(\Pi - NV_{FB}(p) - (1-p)NV'_{FB}(p)) \right) \notag 
    \end{align}
    \begin{equation}
        \label{eqn:time_diffeq}
\left(1 + \frac{Np \lambda}{r}\right)V_{FB}(p) + \frac{Np(1-p)\lambda}{r}V_{FB}'(p) = p\lambda \left( \frac{\Pi}{r} + R \right)   
    \end{equation}

    But \eqref{eqn:time_diffeq} exactly matches the differential equation \eqref{eqn:coop_effort_diffeq}, which $V_{FB}$ was constructed to solve. Hence, $V_{FB}$ satisfies the differential equation on this region. It remains to verify that the efficient policy $k = 1$ is optimal in this region given $V_{FB}$. That is, 
    \begin{align*}
        &p \lambda R_w  + \frac{\lambda}{r} p \pi_w - p \frac{\lambda}{r}  V_{FB}(p) - \frac{\lambda}{r} p(1-p)V'_{FB}(p)  \\
        &=  p \lambda R_w  + \frac{\lambda}{r} p \pi_w - p\lambda\left(\frac{\Pi/N}{r} + R/N \right) +  V_{FB}(p)/N \\
        &=  p \lambda R_w  + \frac{\lambda}{r} p (\pi_w - \pi_s) - \frac{1}{N}p\lambda\left(\frac{\Pi - N\pi_s}{r} + R\right) +  V_{FB}(p)/N \\
        &=  \frac{p}{p_I}\pi_s - \frac{1}{N}\frac{p}{p_{FB}}\pi_s +  V_{FB}(p)/N \\
        &=  \frac{N-1}{N}\frac{p}{p_{FB}}\pi_s+  V_{FB}(p)/N \\
        &> \frac{N-1}{N}\pi_s + \pi_s/N = \pi_s
    \end{align*}
    where the fourth line applied the definitions of $p_{FB}$ and $p_I$, the fifth line used the fact that $p_{FB} = p_I$ when the efficiency condition holds, and the final line uses the fact that $p > p_{FB}$ and $V_{FB}(p) \ge \pi_s$. Hence, the term multiplying $k$ inside the maximization is always positive, and thus it is optimal to set $k = 1$ for $t < t_{FB}$. But that implies that $p(t_{FB}) = p_{FB}$ and so the best-response to the unobservable problem when all the other agents employ a cutoff at $t_{FB}$ is to do the same. Finally, note that in the unobservable problem, the incentive of the agent to exert effort is monotonic in the belief; hence, if other agents do not exert maximum effort on research before $t_{FB}$ the belief of the agent is strictly higher, and the best-response of the agent must still be to exert full effort prior to $t_{FB}$; thus, in any equilibrium, it must be the case that all agents are exerting full effort prior to $t_{FB}$. It now suffices to argue that experimentation must stop at $t_{FB}$ after all agents exerted full effort prior to $t_{FB}$. Suppose not, that some agent selects a strategy exerts effort past that time to $t_T > t_{FB}$. The strategy induces a Bellman value $u$, which is a viscosity solution to the HJB. Take a sequence of $t_n \to t_T$ such that $u_p$ and $u_t$ are defined on the sequence. Since the agent exerts positive effort on $t_n \to t_T$, it must be the case that 
    \[ p(t_n)\frac{\lambda}{r}(\pi_w - u(p(t_n),t_n) - (1 - p(t_n))u_p(p(t_n),t_n)) > \pi_s - p(t_n)\lambda R_w \]
    \[ p(t_n)\frac{\lambda}{r}(\pi_s - u(p(t_n),t_n) - (1 - p(t_n))u_p(p(t_n),t_n)) > \pi_s\left( 1 - \frac{p(t_n)}{p_{I}}\right)\]
    In the limit of the above expression, since experimentation stops at $t_T$, $u(p(t_T), t_T) = \pi_s$, so
    \[ - p_T\frac{\lambda}{r} (1 - p_T)\lim_{t_n \to t_T} u_p(p(t_n),t_n) > \pi_s\left( 1 - \frac{p_T}{p_{FB}}\right) \]
    Since the belief at $t_{FB}$ is $p_{FB}$, the limit belief $p_T$ must be less than $p_{FB} = p_I$. But this implies the right hand side is positive, and so the limit 
    \[\lim_{t_n \to t_T} u_p(p(t_n),t_n) < 0  \]
    But since $u(p_T, t_T) = \pi_s$, this implies that there exists some $p > p_T$ such that $u(p, t_T) < \pi_s$, a contradiction, since each agent can guarantee a flow value of $\pi_s$ by always playing the safe arm, since the efficiency condition implies $\pi_l + r R_l = \pi_s$ (and so losing grants the same flow value as the safe arm). Hence, all agents stop experimenting at $t_{FB}$, so the efficient solution is the unique equilibrium.
\hfill \qedsymbol

\subsection*{Proof of Theorem \ref*{thm:guarantee_2}}

First, I show that the HJB resulting from the sharing rule $(\alpha_I, \alpha_C)$ in the nonattributable breakthrough model results in an HJB for appropriate parameters in the baseline model. Consider an infinitesimal time increment, $[t, t + dt)$. The subjective probability of a breakthrough in this time interval is given by $p K \lambda \ dt$, so the instantaneous flow payoff term in the HJB is 
    \begin{align*} 
    &r \left[ (1-k_i)\pi_s + p K \lambda \left(\alpha_I \frac{k_i}{K} R + (1-\alpha_I) \frac{R}{N}\right) \right] dt  \\
    =&r \left[ (1-k_i)\pi_s + p \lambda \left(k_i \left( \alpha_I R + (1-\alpha_I)\frac{R}{N}\right) +  K_{-i} (1-\alpha_I)\frac{R}{N}\right) \right] dt \\
    =&r \left[ (1-k_i)\pi_s + p \lambda \left(k_i \tilde{R}_w +  K_{-i} \tilde{R}_l \right) \right] dt 
    \end{align*}
    where $\tilde{R}_w = \alpha_I R + (1-\alpha_I)\frac{R}{N}$ and $\tilde{R}_l = (1-\alpha_I)\frac{R}{N}$. The continuation terms in the payoff are then 
    \begin{align*}
    &p K \lambda \left[ \alpha_C \frac{k_i}{K} \Pi + (1-\alpha_C) \frac{\Pi}{N} - u(p) - (1-p)u'(p) \right] dt  \\
    =&p \lambda \left[  k_i \left( \alpha_C \Pi  + (1-\alpha_C) \frac{\Pi}{N} \right) + K_{-i} (1-\alpha_C) \frac{\Pi}{N} - (k_i + K_{-i})u(p) - (k_i + K_{-i})(1-p)u'(p) \right] dt\\
    =&p \lambda \left[  k_i \tilde{\pi}_w + K_{-i} \tilde{\pi_l} - (k_i + K_{-i})u(p) - (k_i + K_{-i})(1-p)u'(p) \right] dt
    \end{align*}
    where $ \tilde{\pi}_w = \alpha_C \left( \Pi \right) + (1-\alpha_C) \frac{ \Pi }{N} $ and $ \tilde{\pi}_l = (1-\alpha_C) \frac{ \Pi }{N} $. Putting these together, the HJB becomes 
    \begin{align*}
        u(p) = \max_{k_i}&\left[ (1-k_i)\pi_s + p \lambda \left(k_i \tilde{R}_w +  K_{-i} \tilde{R}_l \right) \right.\\
        &\left.+ p \frac{\lambda}{r} \left[  k_i \tilde{\pi}_w + K_{-i} \tilde{\pi_l} - (k_i + K_{-i})u(p) - (k_i + K_{-i})(1-p)u'(p) \right] \right]
    \end{align*}
    Rearranging, we get 
    \begin{align*}
        u(p) =& \pi_s + K_{-i}\left( p\lambda \tilde{R}_l + p\frac{\lambda}{r}(\tilde{\pi}_l - u(p) - (1-p)u'(p))\right) \\
        &+  \max_{k_i}\left[k_i\left( p \frac{\lambda}{r} \left[  \tilde{\pi}_w - u(p) - (1-p)u'(p) \right]  -\left(\pi_s - p \lambda \tilde{R}_w \right) \right)\right]
    \end{align*}
    which matches the HJB from the baseline model in \eqrefb{eqn:hjb_baseline}, with payoff parameters $\tilde{R}_w, \tilde{R}_l, \tilde{p}_w, \tilde{p}_l$. Further, note that 
    \[ G(c^K_{\alpha_I, \alpha_C}) =  r\left(1-\alpha_I\right)\frac{R}{N}  + \left(1-\alpha_C\right)\frac{\Pi}{N} = r\tilde{R}_l + \tilde{\pi}_l \]
    Then 
    \[ G(c^K_{\alpha_I, \alpha_C}) = \pi_s \iff \frac{\pi_s - \tilde{\pi}_l}{r} = \tilde{R}_l \]
    So the result follows from Theorem \ref*{thm:noncoop_efficient}. \hfill \qedsymbol

\subsection*{Proof of Proposition \ref*{prop:hetero_cooperative}}
    Note that since the total instantaneous and continuation payoffs are fixed at $R$ and $\Pi$ respectively, the total payoff to all agents after any outcome is 
\[ \int^\tau re^{-rt}\pi_s(M - K) \ dt + re^{-r\tau}R + e^{-r\tau}\Pi. \]
Note that this is identical to the derivation of equation \eqrefb{eqn:hjb_cooperative}, but with $M$ instead of $N$. Hence, the rest of the result follows from the proof of Theorem \ref*{thm:cooperative} replacing $N$ with $M$. Specifically, the HJB for the total value function is 
\begin{align*}
    V^H(p) &= M\pi_s + \max_K \left[K\left(p\frac{\lambda}{r}\left(\Pi - V^H(p) - (1-p)(V^H)'(p) \right) - (\pi_s - p \lambda R) \right) \right].
\end{align*}
and the value function is 
\begin{equation}\label{eqn:hetero_value_function_fb}
V^H_{FB}(p) = \begin{cases}
 M\pi_s & p < p_{FB} \\
 Mp \frac{\lambda \left( \frac{1}{r}\Pi + R \right)}{1 + \frac{M\lambda}{r}} + M\frac{\pi_s(1 - p^H_{FB})\frac{M\lambda}{r}}{\left(1 + \frac{M\lambda}{r}\right) \phi(p_{FB})}\phi(p) & p \ge p_{FB}.
\end{cases}
\end{equation}
\hfill \qedsymbol    

\subsection*{Proof of Proposition \ref*{prop:hetero_noncoop}}
I recreate the similar steps as in the proof of Theorem \ref*{thm:noncoop_efficient}, and prove some lemmas characterizing analogues of $p_\times$ and $p_I$ for the heterogenous resource setting. 
Taking the same approach to writing out the individual HJB for the best-response problem of agent $i$,
    \begin{align*}
    u_i(p) &= \mu_i \pi_s + K_{-i}(p)\left(p\lambda R_{l,i} + b_{I,i}(p,u,u') - p\frac{\lambda}{r}(\pi_{w,i} - \pi_{l,i}) \right) +  \max_{k_i} \left[k_i \left( b_{I,i}(p, u, u') - c_{I,i}(p) \right)  \right]
\end{align*}
where 
\[ b_{I,i}(p, u, u') = p\frac{\lambda}{r}(\pi_{w,i} - u_i(p) - (1-p)u'_i(p)) \]
and 
\[     c_{I,i}(p) = \pi_s - p \lambda R_{w,i}. \]
Following the same characterization trick as in Lemma \ref*{lem:baseline_br}, the optimal best-response can be characterized as follows:

\begin{lemma}\label{lem:hetero_br}
The best-response of agent $i$ takes the form
    \[k_i = \begin{cases}
    0 & u(p) < \mu_i\pi_s + K_{-i}(p)\left(\pi_s - p\lambda (R_{w,i} - R_{l,i}) - \frac{p\lambda}{r}\left( \pi_{w,i} - \pi_{l,i} \right) \right) \\
    \in [0,\mu_i] & u(p) = \mu_i\pi_s + K_{-i}(p)\left(\pi_s - p\lambda (R_{w,i} - R_{l,i})- \frac{p\lambda}{r}\left( \pi_{w,i} - \pi_{l,i} \right) \right)  \\
    \mu_i & u(p) > \mu_i\pi_s + K_{-i}(p)\left(\pi_s - p\lambda (R_{w,i} - R_{l,i}) - \frac{p\lambda}{r}\left( \pi_{w,i} - \pi_{l,i} \right) \right)  \\
    \end{cases}.\]
\end{lemma}

The level curves can be analogously defined: 
\[ \mathcal{D}_{K_{-i}, i} = \left\{(p,u) \in [0,1]\times \mathbb{R}_+ \left\vert \, u = \mu_i\pi_s + K_{-i}(p)\left(\pi_s - p\lambda (R_{w,i} - R_{l,i}) - \frac{p\lambda}{r}\left( \pi_{w,i} - \pi_{l,i} \right) \right)\right.\right\}. \] 
This implies that the threshold where the level curves intersect is given by 
\[ p^H_{\times, i} = \frac{\pi_s/\lambda}{R_{w,i} - R_{l,i} + \frac{1}{r}\left(\pi_{w,i} - \pi_{l,i} \right)}. \]
Using the normalization condition, $R_{w,i} = R - \sum_{j \neq i} R_{l,j}$ and  $\pi_{w,i} = \Pi - \sum_{j \neq i} \pi_{l,j}$, so 
\[ p^H_{\times, i} = \frac{\pi_s/\lambda}{R - \sum_j R_{l,j} + \frac{1}{r}\left(\Pi - \sum_j \pi_{l,j} \right)}. \]
Note that this implies that $p^H_{\times, i} = p^H_{\times, j}$, so I drop the $i$ subscript. With some more algebra, 
\[ p^H_{\times} = \frac{\pi_s/\lambda}{R + \frac{1}{r}(\Pi - M\pi_s) + \sum_i \left(\frac{1}{r}\left(\mu_i \pi_s -  \pi_{l,i}\right) - R_{l,i}\right)}. \]
Define 
\begin{equation}\label{eqn:def_delta}
\delta_i = \frac{1}{r}\left(\mu_i \pi_s -  \pi_{l,i}\right) - R_{l,i}. 
\end{equation}
Then $p^H_\times$ can be rewritten as 
\begin{equation}\label{eqn:ptimes_hetero}
    p^H_{\times} = \frac{\pi_s/\lambda}{R + \frac{1}{r}(\Pi - M\pi_s) + \sum_i \delta_i }.
\end{equation} 

\begin{lemma}\label{lem:hetero_ptimes_bound}
    Suppose that $p^H_\times > p^H_{FB}$. Then the efficient solution cannot be an MPE.
\end{lemma}
\begin{proof}
    Consider the best-response problem of agent $i$, when other agents use cutoff strategies with cutoff at $p^H_{FB}$. Suppose, for sake of contradiction, that a cutoff strategy at $p^H_{FB}$ is a best response. By the boundary condition of the HJB, the value function at the cutoff must be $\mu_i \pi_s$, so $u_i\left(p^H_{FB}\right) = \mu_i \pi_s$. But the point $(p^H_{FB}, \mu_i \pi_s)$ lies strictly below the half-plane bounded above by $\mathcal{D}_{M-\mu_i, i}$, since the curve $\mathcal{D}_{M- \mu_i, i}$ is a line with negative slope passing through $(p^H_\times, \mu_i \pi_s)$ and $p^H_\times > p^H_{FB}$. By the contradiction hypothesis and continuity of any HJB solution, there exists some interval $(p^H_{FB}, p^H_{FB} + \epsilon)$ such that $\epsilon > 0$ and $i$ exerts full effort $k_i = \mu_i$ on this interval, but the value function $u_i$ lies below $\mathcal{D}_{M - \mu_i, i}$. But this is a contradiction of Lemma \ref{lem:hetero_br}. Hence a cutoff at $p^H_{FB}$ cannot be a best response, so the efficient solution is not an MPE. 
\end{proof}

Now, we generalize the analogous individual experimentation threshold as $p_I$, (i.e., where experimentation would stop if $K_{-i} = 0$ everywhere). The analogue is given by 
\begin{align}
p^H_{I, i} &= \frac{\pi_s / \lambda }{R_{w,i} + \frac{1}{r}(\pi_{w,i} - \mu_i\pi_s)} \notag \\ 
&= \frac{\pi_s / \lambda }{R - \sum_{j \neq i} R_{l,j} + \frac{1}{r}\left(\Pi - \sum_{j \neq i} \pi_{l,j} - \mu_i\pi_s\right)}\notag \\
&= \frac{\pi_s / \lambda }{R  + \frac{1}{r}\left(\Pi - M\pi_s\right) + \sum_{j \neq i}\left(\frac{1}{r}\left(\mu_j \pi_s - \pi_{l,j} \right) - R_{l,j} \right)}\notag \\
&= \frac{\pi_s / \lambda }{R  + \frac{1}{r}\left(\Pi - M\pi_s\right) + \sum_{j \neq i}\delta_j}. \label{eqn:pI_hetero}
\end{align}

\begin{lemma}\label{lem:hetero_pI_bound}
    Suppose that all agents $j \neq i$ have stopped experimenting at $p_T > p^H_{I,i}$ and at all beliefs below $p_T$. Then agent $i$ cannot stop experimenting at $p_T$. 
\end{lemma}
\begin{proof}
Suppose, for sake of contradiction, that agent $i$ stops experimenting at $p_T$. Then the value function $u_i(p) = \mu_i \pi_s$ for $p < p_T$ by the boundary condition. By the boundary conditions and the HJB, at any point $p \in [p^H_{I,i}, p_T]$, the HJB indicates that
    \begin{align*} 
    \mu_i\pi_s = &\mu_i\pi_s + \max_{k_i} \left[k_i \left( p \frac{\lambda}{r}(\pi_{w,i} - \mu_i\pi_s) - (\pi_s - p \lambda R_{w,i}) \right)  \right] \\
    0 =&  \max_{k_i} \left[k_i \left( p\frac{\pi_s}{p^H_{I,i}}  - \pi_s \right)  \right].
    \end{align*}
But this is a contradiction; since $p > p^H_{I,i}$, the maximal $k_i$ is $\mu_i$, and hence the RHS here cannot be zero, but $p_T > p$. Hence, agent $i$ cannot stop experimentation at $p_T$.
\end{proof}

Further, I extend Lemma \ref*{lem:stopping_lower_bound} to this setting. 
\begin{lemma}\label{lem:hetero_stopping_lower_bound}
Suppose $p^H_\times \le p^H_{I,i}$. Then no agent can experiment at any belief at or below $p^H_\times$. 
\end{lemma}
\begin{proof}
Suppose agent $i$ was experimenting at beliefs down to $p_T < p^H_\times$. Let $u_i$ be the value function of that agent. Since experimentation stops at $p_T$, $u_i(p_T) = \mu_i\pi_s$. Since $u_i$ must be a viscosity solution to the best-response HJB for some $K_{-i}(p)$, we can take a sequence $p_n \to p_T$ such that $u'_i(p_n)$ is well defined, $p_n > p_T$. Since the agent was experimenting, $k_i(p_n) > 0$, so it must have been the case that 
\[ p_n\frac{\lambda}{r} \left(\pi_{w,i} - u_i(p_n) - (1-p_n)u'_i(p_n)\right) > \pi_s - p_n \lambda R_{w,i}\]
\[ p_n \left(\lambda R_{w,i} + \frac{\lambda}{r} \left(\pi_{w,i} - u_i(p_n) - (1-p_n)u'_i(p_n)\right)\right) > \pi_s\]
\[ p_n \left(\frac{\pi_s}{p^H_{I,i}}+ \frac{\lambda}{r} \left(\mu_i\pi_s - u_i(p_n) - (1-p_n)u'_i(p_n)\right)\right) > \pi_s\]
\[ p_n \left(\frac{\lambda}{r} \left(\mu_i\pi_s - u_i(p_n) - (1-p_n)u'_i(p_n)\right)\right) > \pi_s\left(1 - \frac{p_n}{p^H_{I,i}} \right)\]
Taking the limit as $p_n \to p_T$, we get 
\[ - \frac{\lambda}{r} p_T(1-p_T)u'_{i,+}(p_T) > \pi_s\left(1 - \frac{p_T}{p^H_{I,i}} \right)\]
where $u'_{i,+}$ denotes the right derivative, since $u$ need not be differentiable at $p_T$. Once again, the right hand side is positive because $p_T < p^H_\times \le p^H_{I,i}$. Therefore, it must be the case that $u'_+(p_T) < 0$. Hence, there must be some point $p \in [p_T, p^H_\times]$ such that $u_i(p) < \mu_i \pi_s$. But this implies that the point $p, u_i(p)$ lies below $\mathcal{D}_{K_{-i}, i}$ for all $K_{-i} \in [0, N-1]$ (since every $\mathcal{D}_{\cdot, i}$ passes through $(p^H_\times, \mu_i\pi_s)$ with nonpositive slope) and an equilibrium exists where agent $i$ is exerting a positive amount of effort on research at that point, a contradiction of Lemma \ref{lem:hetero_br}. 
\end{proof}

Finally, I prove Proposition \ref*{prop:hetero_noncoop} using the lemmas proved above. First, I show that the efficient solution is an MPE of the noncooperative game if $\frac{\mu_i \pi_s -  \pi_{l,i}}{r} = R_{l,i}$ for all $i$. Recall by equation \eqref{eqn:def_delta}, this is by definition equivalent to $\delta_i = 0$ for all $i$. Examining the definitions of $p^H_\times$ and $p^H_{I,i}$ in equations \eqref{eqn:ptimes_hetero} and \eqref{eqn:pI_hetero}, this implies that $p^H_{FB} = p^H_\times = p^H_{I,i}$ for all $i$. So it suffices to take a verification approach; I construct the value function and check that it is smooth, and increasing above $p^H_{FB}$. Above $p^H_{FB}$, the differential equation implied by the HJB is 
    \begin{equation*}
         u_i(p)+ \frac{M\lambda}{r}p u_i(p) + \frac{M\lambda}{r} p(1-p)u_i'(p) = p\lambda\left( \mu_i R_{w,i} + (M - \mu_i)R_{l,i} + \frac{1}{r}(\mu_i \pi_{w,i} + (M - \mu_i)\pi_{l,i}) \right) 
    \end{equation*}
    The right hand side can be rewritten:
    {\small \begin{align}
         &= p\lambda\left( \mu_i\left(R - \sum_{j \neq i} R_{l, j} \right) + (M - \mu_i)R_{l,i} + \frac{1}{r}\left(\mu_i \left(\Pi - \sum_{j \neq i} \pi_{l, j} \right) + (M - \mu_i)\pi_{l,i}\right) \right) \notag \\
         &= p\lambda\left( \mu_i R - \mu_i \sum_{j} R_{l, j} + MR_{l,i} + \frac{1}{r}\left(\mu_i \Pi - \mu_i \sum_{j} \pi_{l, j} + M\pi_{l,i}\right) \right) \notag \\
         &= p\lambda\left( \mu_i \left( R + \frac{1}{r}\Pi\right)  - \mu_i \sum_{j} \left( R_{l, j} + \frac{1}{r}\pi_{l,j}\right) + M\left( R_{l,i} + \frac{1}{r}\pi_{l,i} \right) \right) \notag \\
         &= p\lambda\left( \mu_i \left( R + \frac{1}{r}\Pi\right)  - \mu_i \sum_{j} \left( \frac{1}{r}\mu_j \pi_s\right) + M\left( \frac{1}{r}\mu_i \pi_s \right) \right) \notag \\
         &= p\lambda\left( \mu_i \left( R + \frac{1}{r}\Pi\right)  - \frac{1}{r}\mu_i M \pi_s + \frac{1}{r}M \mu_i \pi_s \right) \notag \\
         &= p\lambda \mu_i \left( R + \frac{1}{r}\Pi\right). \notag 
    \end{align} }
    where the third step uses the assumption that $\frac{\mu_i \pi_s -  \pi_{l,i}}{r} = R_{l,i}$ for all $i$.
    So the differential equation becomes 
    \begin{equation}\label{eqn:hetero_diffeq}
        u_i(p)+ \frac{M\lambda}{r}p u_i(p) + \frac{M\lambda}{r} p(1-p)u_i'(p) =p\lambda \mu_i \left( R + \frac{1}{r}\Pi\right).
    \end{equation}
    Take the functional analogous to the Theorem \ref*{thm:cooperative} solution, but replacing $N$ with $M$:
\[ \phi_M(p) = (1-p)\Omega(p)^{\frac{r}{M\lambda}}. \]
Note that 
\[ \phi_M(p) + \frac{M\lambda}{r}p\phi_M(p) + \frac{M\lambda}{r}p(1-p)\phi_M'(p) = 0. \]\
So solutions to the differential equation \eqref{eqn:hetero_diffeq} have the form:
\[ u_i(p) = p\lambda\frac{\mu_i \left( R + \frac{1}{r}\Pi\right)}{1 + M\lambda/r} + C\phi_M(p)\]
Value matching at $p^H_{FB} = \mu_i \pi_s$ pins down the constant $C$: 
\begin{align}
    C &= \mu_i \frac{\pi_s}{\phi_M(p^H_{FB})} - p^H_{FB}\lambda\frac{\mu_i\left(R + \frac{\Pi}{r} \right) }{(1 + M\lambda/r)\phi_M(p^H_{FB})} \notag \\
    &= \mu_i\frac{\pi_s(1 + M\lambda/r) - p^H_{FB}\lambda\left(R + \frac{\Pi}{r} \right) }{(1 + M\lambda/r)\phi_M(p^H_{FB})} \notag \\
    &= \mu_i\frac{\pi_s(1 + M\lambda/r) - \pi_s\left(1 + \frac{M \lambda p^H_{FB}}{r} \right) }{(1 + M\lambda/r)\phi_M(p^H_{FB})}\notag \\
    &= \mu_i\frac{\pi_s(M\lambda/r)(1 - p^H_{FB}) }{(1 + M\lambda/r)\phi_M(p^H_{FB})}
\end{align} 
Then the value function becomes exactly:
\[ u_i(p) = \frac{\mu_i}{M}V^H_{FB}(p) \]
where $V^H_{FB}$ was given from \eqref{eqn:hetero_value_function_fb}. The standard verification argument shows that indeed, when $\frac{\mu_i \pi_s -  \pi_{l,i}}{r} = R_{l,i}$ for all $i$, the efficient solution is an MPE.

Now, I argue that if the efficient solution is an MPE, $\frac{\mu_i \pi_s -  \pi_{l,i}}{r} = R_{l,i}$ for all $i$. Recall the definition \eqref{eqn:def_delta}: 
\[ \delta_i = \frac{\mu_i \pi_s -  \pi_{l,i}}{r} - R_{l,i}\]
From equation \eqref{eqn:ptimes_hetero} and Lemma \ref{lem:hetero_ptimes_bound}, if $\sum_i \delta_i < 0$, the efficient solution cannot be an MPE. Hence, if the efficient solution is an MPE, it must be the case that 
\[ \sum_i \delta_i \ge 0. \]
Further, by Lemma \ref{lem:hetero_pI_bound}, if $p^H_{I,i} < p^H_{FB}$ for some $i$, the efficient solution cannot be an MPE because agent $i$ cannot stop experimenting at $p^H_{FB}$. So the efficient solution being an MPE implies that $p^H_{I,i} \ge p^H_{FB}$. Examining equation \eqref{eqn:pI_hetero}, this implies that 
\[\sum_{j \neq i} \delta_j \le 0. \]
for all $i$. Since the sum of all $\delta_i$'s must be nonnegative, this implies that each $\delta_i$ must be nonnegative. But then the only way the above inequality can be satisfied is if $\delta_i = 0$ for all $i$. Hence, if the efficient solution is an MPE, $\delta_i = 0$ for all $i$, so $\frac{\mu_i \pi_s -  \pi_{l,i}}{r} = R_{l,i}$ for all $i$. 

Finally, I show that if the condition is satisfied, the efficient MPE is the unique solution. When the condition is satisfied, $p^H_{FB} = p^H_\times = p^H_{I,i}$ for all $i$, so by Lemma \ref{lem:hetero_stopping_lower_bound}, all agents must stop experimenting at or above $p^H_{FB}$. To show that no agent drops effort above $p^H_{FB}$, suppose $k_i(p) < \mu_i$ for $p > p^H_{FB}$. Since $rR_{l,i} + \pi_{l,i} = \mu_i\pi_s$, agent $i$ could obtain $\mu_i \pi_s$ by playing $k_{-i} = 0$, and so $u_i(p) \ge \mu_i\pi_s$. If $u_i(p) > \mu_i\pi_s$, then $(p, u_i(p))$ lies above $\mathcal{D}_{K_{-i}, i}$ for any $K_{-i}$ so $k_i(p) < \mu_i$ is a contradiction of Lemma 1. If $u_i(p) = \mu_i \pi_s$, then by Lemma 1, the only way $k_i(p) < \mu_i$ can be an equilibrium best-response policy requires $K_{-i}(p) = 0$, so the HJB implies that in order for $k_i < 1$ to be optimal, taking any sequence $p_n < p$, $p_n \to p$ and $u_i$ differentiable at $p_n$, we have  
\[ p_n\frac{\lambda}{r}(\pi_{w,i} - u_i(p_n) - (1-p_n)u_i'(p-n)) \le (\pi_s - p_n \lambda R_{w,i})  \]
In the limit as $p_n \to p$, 
\[ -p\frac{\lambda}{r}(1-p)u'_{i,-}(p) \le \pi_s \left(1 - \frac{p}{p^H_{I,i}}\right)  \]
where $u'_{i,-}$ is the left-derivative (as $u_i$ need not be differentiable). But this implies that the left-derivative of $u_i(p)$ is positive (since $p > p^H_{FB}=p^H_{I,i}$ the right-hand side is negative) and so there must exist some $p' < p$, such that $u_i(p') < \mu_i\pi_s$, a contradiction of the fact that any agent can guarantee at least $\mu_i\pi_s$ by always playing $k_i = 0$. Hence in either case of $u_i(p) = \mu_i\pi_s$ or $u_i(p) > \mu_i\pi_s$ it cannot be an equilibrium best-response to play $k_i(p) < \mu_i$ at $p$, and so the only equilibrium must be the efficient solution as each agent must set $k_i(p) = \mu_i$ above $p^H_{FB}$.
\hfill \qedsymbol    

\subsection*{Proof of Proposition \ref*{prop:hetero_contracts}}
    Using the identical arguments as Theorems \ref*{thm:guarantee} and \ref*{thm:guarantee_2}, the contract induces a game with 
    \[ \tilde{\pi}_{l,i} = \Pi(1-\alpha_C)\frac{\mu_i}{M}, \quad \tilde{R}_{l, i} = \Pi(1-\alpha_I)\frac{\mu_i}{M} \]
    so then the efficiency condition from Proposition \ref*{prop:hetero_noncoop}, becomes
    \begin{align*}
         g\left(c^{\cdot, H}_{\alpha_I, \alpha_C}\right)  = \pi_s \iff & rR\left(1-\alpha_I\right)\frac{1}{M} + \Pi\left(1-\alpha_C\right)\frac{1}{M} = \pi_s \\
        \iff & r\tilde{R}_{l,i} + \tilde{\pi}_{l,i} = \mu_i \pi_s \\
        \iff & \frac{\mu_i \pi_s - \tilde{\pi}_{l,i}}{r} = \tilde{R}_{l, i}
    \end{align*}
\hfill \qedsymbol    

\section{Inefficient Equilibria}
\label{sec:inefficient}
In this analysis, I focus on symmetric equilibria in weakly monotonic strategies here.\footnote{This restriction avoids the asymmetric switching equilibria seen in \cite{krc2005}, which are inefficient anyhow, either in the amount of experimentation or the rate of experimentation (or both). For expanded discussion of other equilibria, see Appendix \ref{app:other_equilibria}.}

\subsection{Undercompetition}
First, I consider the case where $\frac{\pi_s - \pi_l}{r} < R_l $. The following result generalizes the main insights from the symmetric equilibrium analysis of \cite{krc2005}, but the approach to characterizing equilibria here is relatively standard and hence my discussion here is brief.

\begin{proposition}[Undercompetition]\label{prop:undercompetitive}
Suppose $\frac{\pi_s - \pi_l}{r} < R_l $. There is a unique symmetric MPE, and in this MPE, agents use weakly monotonic strategies and experimentation stops at $p_I$. 
\end{proposition}
\begin{proof}
    Since $p_\times > p_I$, $(p_T, \pi_s)$ lies below $\mathcal{D}_{N-1}$, and since we are considering symmetric equilibri, Lemma \ref*{lem:baseline_br} implies that just before $p_T$, total effort cannot have been $N$, so each agent was exerting an interior amount of effort. This implies that the $b_I = c_I$ with equality in this region, and so $u$ is pinned down by the differential equation 
    \begin{equation}\label{eqn:interior_diffeq}
         p u(p) + p(1-p)u'(p) = p \left( r R_w + \pi_w \right) - \frac{r \pi_s}{\lambda} 
    \end{equation} 
    This has the strictly convex solution: 
    \begin{equation*}
        W(p) = \left(r R_w + \pi_w - \frac{r}{\lambda}\pi_s \right) -\frac{r\pi_s}{\lambda}\varphi(p) + C(1-p)
    \end{equation*}
    where 
    \[ \varphi(p) = (1-p)\ln \left( \frac{1-p}{p}\right) \]
    and $C$ is some constant. Note the choice of $C$ determines which $p_T$ satisfies $W(p_T) = \pi_s$. Further, since any choice of $p_T < p_I$ implies that $W(p) > V_{FB}(p)$ at some $p$, and $V_{FB}$ by construction is an upper bound on the average payoff and the equilibrium is symmetric, it must be the case that $p_T = p_I$, which implies that the constant $C^*$ satisfies:
    \begin{equation*}
        C^* = \frac{1}{1-p_I}\left[\left(1 + \frac{r}{\lambda}\right)\pi_s - rR_w - \pi_w +\frac{r\pi_s}{\lambda}\varphi(p_I) \right]
    \end{equation*}
    Then $W^* = \left(r R_w + \pi_w - \frac{r}{\lambda}\pi_s \right) -\frac{r\pi_s}{\lambda}\varphi(p) + C^*(1-p)$ is such that $W^*(p_I) = \pi_s$. Plugging into the differential equation \eqref{eqn:interior_diffeq}, $(W^*)'(p_I) = 0$. Then below $p_I$ the value function is constant at $\pi_s$. Note now that $W^*$ intersects $\mathcal{D}_{N-1}$ at some $p^\dagger$ satisfying the following implicit equation:
    \begin{equation}\label{eqn:switching_belief}
        W^*(p^\dagger) = \pi_s + (N-1)\left(\pi_s - p^\dagger\lambda (R_w - R_l)- \frac{p^\dagger\lambda}{r}\left( \pi_w - \pi_l \right) \right)
    \end{equation}
    Then at every point in $[p_I, p^\dagger]$, the equilibrium effort level corresponds to which $\mathcal{D}$ surface $p, W^*(p)$ lies on: 
    {\small \[ k^\dagger(p) = \frac{1}{N-1} \left(\frac{W^*(p) - \pi_s}{\pi_s - p\lambda (R_w - R_l)- \frac{p\lambda}{r}\left( \pi_w - \pi_l \right) }\right) \] }
    Above $p^\dagger$, the value function satisfies the cooperative equation \eqref{eqn:coop_value_function_effort}, with constant chosen for continuity. Let $V^*$ be such a value function. To finish the verification, we need to check for differentiability at $p^\dagger$. That is, 
    {\small \begin{align*}
        N\frac{\lambda}{r} p^\dagger(1-p^\dagger)(W^*)'(p^\dagger) =& Np^\dagger\frac{\lambda}{r} (rR_w + \pi_w) - Np^\dagger \frac{\lambda}{r} W^*(p^\dagger) - N\pi_s \\
        =& N p^\dagger\frac{\lambda}{r} (rR_w + \pi_w) - \left( 1 + N p^\dagger\frac{\lambda}{r} \right) W^*(p^\dagger) + W^*(p^\dagger) - N\pi_s \\
        =& N p^\dagger\frac{\lambda}{r} (rR_w + \pi_w) - \left( 1 + N p^\dagger\frac{\lambda}{r} \right) W^*(p^\dagger) - (N-1)p^\dagger\lambda (R_w - R_l) \\
        &- (N-1)\frac{p^\dagger\lambda}{r}\left( \pi_w - \pi_l \right) \\
        =& p^\dagger\frac{\lambda}{r} (r(R_w + (N-1)R_l) + \Pi) - \left( 1 + N p^\dagger\frac{\lambda}{r} \right) W^*(p^\dagger) \\
        =& p^\dagger\frac{\lambda}{r} (rR  + \Pi) - \left( 1 + N p^\dagger\frac{\lambda}{r} \right) V^*(p^\dagger) \\
        =& N\frac{\lambda}{r} p^\dagger(1-p^\dagger)(V^*)'(p^\dagger) 
    \end{align*} }
    where in the second step we added and subtracted a $W^*(p^\dagger)$, the third step used \eqref{eqn:switching_belief}, the fifth step used that $V^*(p^\dagger) = W^*(p^\dagger)$, and the final step follows from the differential equation \eqref{eqn:coop_effort_diffeq} pinning down $V^*$. Hence, the constructed value function is continuous and differentiable, completing the verification. Since the differential equations pinned down a unique solution, this is the only symmetric equilibrium.

\end{proof}

Note that in the symmetric equilibrium from Proposition \ref*{prop:undercompetitive}, the strategies used are not cutoff strategies; that is, prior to the end of experimentation, the agents reduce effort gradually towards zero. 

The argument follows the same ideas as the proof of Proposition 5.1 in \cite{krc2005}.\footnote{For the reader familiar with the economics bandit literature, the model of \cite{krc2005} is analogous to my model with $\pi_w = \pi_l = \lambda h$, $R_w = h$ and $R_l = 0$.} Since $\pi_l > \pi_s$ implies $p_{FB} < p_I < p_\times$, the value function must cross $\mathcal{D}_{N-1}$ at some $p \in [p_{FB}, p_\times]$. Above $\mathcal{D}_{N-1}$, all agents exert full effort, and thus the HJB generates a differential equation governing the law of motion there. Below $\mathcal{D}_{N-1}$, in any symmetric MPE, all agents must be exerting an interior amount of effort $\in (0,1)$ and hence the condition for $k_i$ to be interior in the HJB provides another differential equation for the value function between $\mathcal{D}_0$ and $\mathcal{D}_{N-1}$. From these two conditions, I use a verification approach and explicitly construct a solution using smooth pasting and value matching conditions. 

Recall that $\frac{\pi_s - \pi_l}{r} < R_l $ implies that $p_I > p_{FB}$, by Lemmas \ref*{lem:ptimes_pFB_comparison} and \ref*{lem:pI_characterization}. That is, in equilibrium, in this case agents experiment less than a social planner would due to the presence of a free-riding effect; that is, losers still benefit from a breakthrough, and so there is an incentive to free-ride on others' effort.

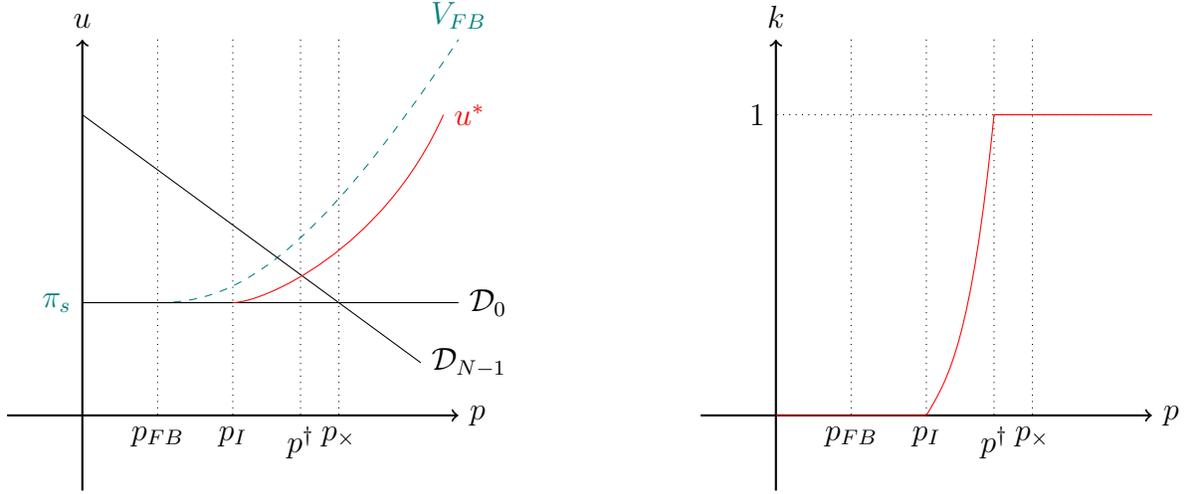
\begin{figure}
    \centering
    \begin{subfigure}[b]{0.45\textwidth}
         \centering
         \begin{tikzpicture}
            \draw[->, thick] (-1,0) -- (5,0) node[anchor=west]{$p$};
            \draw[->, thick] (0,-1,0) -- (0,5) node[anchor=south]{$u$};
            \draw[dotted] (1,5) -- (1,0) node[anchor=north]{$p_{FB}$};
            \draw[teal] (2,1.5) -- (0,1.5) node[anchor=east]{$\pi_s$};
            \draw[teal, dashed] (1,1.5) .. controls (2,1.5) and (3,2) .. (5,5) node[anchor=south]{$V_{FB}$};
            \draw (0,1.5) -- (5,1.5) node[anchor=west]{$\mathcal{D}_0$};
            \draw (0,4) -- (4.5,0.7) node[anchor=west]{$\mathcal{D}_{N-1}$};
            \draw[dotted] (2,5) -- (2,0) node[anchor=north]{$p_I$};
            \draw[dotted] (2.9,5) -- (2.9,0) node[anchor=north]{$p^\dagger$};
            \draw[dotted] (3.41,5) -- (3.41,0) node[anchor=north]{$p_\times$};
            \draw[red] (2,1.5) .. controls (2.4,1.5) and (4,2.2) .. (4.8,4) node[anchor=west]{$u^*$};
        \end{tikzpicture}
        \caption{Value functions. The dashed line $V_{FB}$ denotes the value function of the first-best solution, and $u^*$ denotes the value function of the agent in the noncooperative equilibrium. }
         \label{fig:undercompetitive_value}
    \end{subfigure}
    \hfill
    \begin{subfigure}[b]{0.45\textwidth}
         \centering
         \begin{tikzpicture}
            \draw[->, thick] (-1,0) -- (5,0) node[anchor=west]{$p$};
            \draw[->, thick] (0,-1,0) -- (0,5) node[anchor=south]{$k$};
            \draw[dotted] (1,5) -- (1,0) node[anchor=north]{$p_{FB}$};
            \draw[dotted] (2,5) -- (2,0) node[anchor=north]{$p_I$};
            \draw[dotted] (2.9,5) -- (2.9,0) node[anchor=north]{$p^\dagger$};
            \draw[dotted] (3.41,5) -- (3.41,0) node[anchor=north]{$p_\times$};
            \draw[dotted] (5,4) -- (0,4) node[anchor=east]{1};
            \draw[red] (5,4) -- (2.9,4);
            \draw[red] (2.9,4) .. controls (2.6, 1) and (2.3, 0.5) .. (2,0);
            \draw[red] (2,0) -- (0,0);
        \end{tikzpicture}
        \caption{Strategies. Note that the unique symmetric equilibrium strategies are weakly monotonic, but not cutoff; effort starts decreasing at $p^\dagger$, and reaches zero at $p_I$. }
         \label{fig:undercompetitive_strat}
    \end{subfigure}
    \caption{Equilibrium value function and symmetric strategy in an undercompetitive experimentation game.}
    \label{fig:undercompetitive}
\end{figure}
The solution exhibits a key features of the \cite{krc2005} symmetric equilibrium (that is, agents taper their effort as the belief approaches $p_I$). Figure \ref{fig:undercompetitive} shows the agent strategies and value function in an example of such an undercompetitive equilibrium.

\subsection{Overcompetition}
Now, consider the case where $\frac{\pi_s - \pi_l}{r} > R_l $. The following characterizes the equilibria in weakly monotonic strategies.

\begin{proposition}[Overcompetition]\label{prop:overcompetitive}
    Suppose $\frac{\pi_s - \pi_l}{r} > R_l $. In any symmetric MPE in weakly monotonic strategies, experimentation stops at some threshold $p_T \in [p_\times, p_I]$. Moreover, for any $p_T$ in $[p_\times, p_I]$, there is a symmetric MPE where all agents use cutoff strategies stopping at $p_T$. 
\end{proposition}
\begin{proof}
    First, I show the second half of the statement; that is, all agents using a cutoff strategy at $p_T$ for any $p_T \in [p_\times, p_I]$ is an equilibrium. It suffices to show that if all other agents are employing a cutoff strategy at $p_T$, the best response is to also use a cutoff strategy at $p_T$. If all other agents are using cutoff strategies, then $K_{-i}$ is $N-1$ at $p > p_T$ and $0$ at $p \le p_T$. At beliefs above $p_T$, the differential equation on the Bellman value dictated by the HJB is then:
    \begin{align*} 
    u(p) = &\pi_s + (N-1) \left[p \lambda R_l + p\frac{\lambda}{r}(\pi_l - u(p) - (1-p)u'(p))\right] \\
    &+  \max_{k_i} \left[k_i \left( p\frac{\lambda}{r}(\pi_w - u(p) - (1-p)u'(p)) - (\pi_s - p \lambda R_w) \right)  \right] 
    \end{align*}
    I use a verification argument in the same manner as in the proof of Theorem \ref*{thm:cooperative}. We want to show that a cutoff strategy at $p_T$ is an optimal policy to this control problem. To do this, I explicitly construct the value function corresponding to this policy and show that this value function is a viscosity solution to the HJB. Since the HJB has a well-behaved control function (the belief law of motion is smooth and Lipschitz), Theorem 2.12 in \cite{bd1997} ensures that a viscosity solution that solves the HJB is exists and is unique.
    
    Note the best response $k_i$ depends on whether
    \[ p\frac{\lambda}{r}(\pi_w - u(p) - (1-p)u'(p))  - (\pi_s - p \lambda R_w) \]
    is positive, negative, or zero. If $k_i = 1$, note that the HJB becomes a differential equation 
    \[  \left(1 + \frac{Np \lambda}{r}\right)u(p) + \frac{Np(1-p)\lambda}{r}u'(p) = \pi_s + \left(Np\frac{\lambda}{r}\left(\frac{\Pi}{N}\right) - \left(\pi_s - p\lambda R\right) \right)   \]
    which reduces to the same differential equation as in the cooperative case, equation \eqref{eqn:coop_effort_diffeq}. As we explicitly solved before, the value function satisfies Equation \eqref{eqn:coop_value_function_effort}:
    \[ V(p) = p \frac{\lambda \left( \frac{\Pi}{r} + R \right)}{1 + \frac{N\lambda}{r}} + C\phi(p)  \]
    Imposing a value-matching condition at $p_T$ to solve for $C$, the value function constructed is:
    \begin{equation*} V(p) = \begin{cases}
 \pi_s & p < p_{T} \\
 p \frac{\lambda \left( \frac{\Pi}{r} + R \right)}{1 + \frac{N\lambda}{r}} + \frac{\pi_s\left(1 + \frac{N\lambda}{r}\right) - p_{T} \lambda \left( \frac{\Pi}{r} + R \right)}{\left(1 + \frac{N\lambda}{r}\right) \phi(p_{T})}\phi(p) & p \ge p_{T}
\end{cases} \end{equation*}
To verify that this is a viscosity solution, we note that this value function is smooth above $p_T$ and satisfies the HJB differential equation in this region (as we constructed).
Below $p_T$, $V$ is also smooth and satisfies the HJB. However, this solution is not differentiable at $p_T$, so we have to check that this value function is both a viscosity subsolution and viscosity supersolution at $p_T$. In particular, $V$ kinks at $p_T$, where its left derivative is 0 and its right derivative can be found from the differential equation from the HJB: 
\begin{align*} 
u'_+(p_T) &= \frac{r}{Np_T(1-p_T)\lambda} \left(Np_T\frac{\lambda}{r}\left(\frac{\Pi}{N} - \pi_s\right) - \left(\pi_s - p_T\lambda R\right) \right)  \\
&= \frac{r}{Np_T(1-p_T)\lambda} \left( p_T\frac{\pi_s}{p_{FB}} - \pi_s \right) < 0
\end{align*}
where we plugged in $p_{FB}$ from \eqrefb{eqn:belief_fb}, and the inequality follows since $p_T \le p_I < p_{FB}$. Hence $V$ kinks downward at $p_T$, so there is no $C^\infty([0,1])$ variation $\phi$ such that $V - \phi$ attains a minimum of $0$ at $p_T$, and so $V$ is trivially a viscosity supersolution. It remains to be shown that $V$ is also a viscosity subsolution. The Hamiltonian expression is 
\begin{align*} H(p, u, Du) =& u(p) -\pi_s - K_{-i}(p)\left[p \lambda R_l + p\frac{\lambda}{r}(\pi_l - u - (1-p)Du )\right] \\
    &-  \max_{k_i} \left[k_i \left( p\frac{\lambda}{r}(\pi_w - u - (1-p)Du) - (\pi_s - p \lambda R_w) \right)  \right]  \end{align*}
Take any $C^\infty([0,1])$ variation $\phi$ such that $V - \phi$ attains a maximum of at $p_T$, and let $\phi(p_T) = \pi_s = V(p_T)$. Then $\phi'(p_T)$ is by construction a superdifferential of $V$ at $p_T$, so $\phi'(p_T) \in \left[ \frac{r}{Np_T(1-p_T)\lambda} \left( p_T\frac{\pi_s}{p_{FB}} - \pi_s \right), 0\right]$. Then at $(p_T, V(p_T), \phi')$, we have that 
\begin{align*} H(p_T, \pi_s, \phi') =&  -  \max_{k_i} \left[k_i \left( p_T\frac{\lambda}{r}(\pi_w - \pi_s - (1-p_T)\phi') - (\pi_s - p_T \lambda R_w) \right)  \right]\\
&\le 0
\end{align*}
and hence $V$ is a viscosity subsolution (implicitly, the form of the Hamiltonian allows for kinks in only one direction). Hence, $V$ is both a subsolution and a supersolution, and so $V$ is a viscosity solution. Since the viscosity solution is unique by Lemma \ref*{lem:viscosity}, $V$ corresponds to the optimal best-response, and so the best-response strategy is also a cutoff strategy at $p_T$. Hence, for any $p_T \in [p_\times, p_I]$, there is a symmetric MPE in cutoff strategies at $p_T$.

Now, I show the first part of the statement, which is that in any MPE in weakly monotonic strategies, experimentation stops in $[p_\times, p_I]$. Lemma \ref*{lem:stopping_upper_bound}, experimentation cannot stop above $p_I$. By Lemma \ref*{lem:stopping_lower_bound}, experimentation cannot stop at $p_T < p_\times$. Hence we are done.
\end{proof}

The rough intuition is as follows: Lemmas \ref*{lem:stopping_upper_bound} and \ref*{lem:stopping_lower_bound} show that experimentation must stop in $[p_\times, p_I]$.
To show the second part of the statement, I explicitly construct the value function corresponding to any potential equilibrium in cutoff strategies at $p_T$ and uses a verification argument to show that the value function is a viscosity solution to the HJB, which must be the unique solution. 

Recall that $p_I < p_{FB}$ when $\frac{\pi_s - \pi_l}{r} > R_l $ by Lemmas \ref*{lem:ptimes_pFB_comparison} and \ref*{lem:pI_characterization}. Thus, Proposition \ref{prop:overcompetitive} implies that the agents experiment past the point where a social planner would, and so the environment is overcompetitive. Intuitively, since the condition $\frac{\pi_s - \pi_l}{r} > R_l $ implies that the payoff loss from losing outweighs the instantaneous compensation to the losers, agents have a ``fear of missing out'' on a potential discovery. In particular, at beliefs in $(p_I, p_{FB}]$, the social planner would rather everyone drop the research project, but agents still experiment because of the winner advantage combined with the negative payoff implication from losing. This gives a continuum of coordination equilibria; that is, for beliefs in the range $[p_\times, p_I]$, it is a best response to quit research if everyone else also quits at that belief.

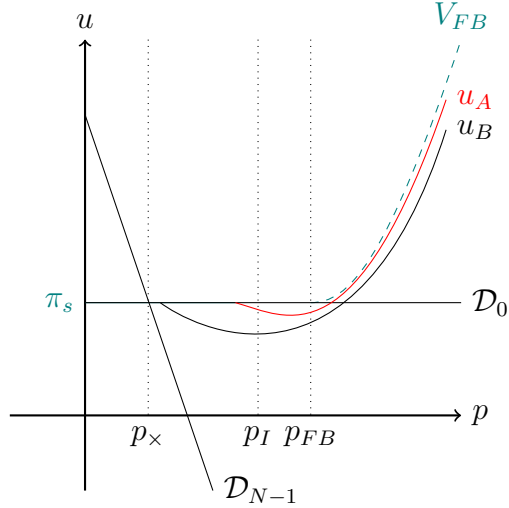
\begin{figure}
    \centering
    \begin{tikzpicture}
    \draw[->, thick] (-1,0) -- (5,0) node[anchor=west]{$p$};
    \draw[->, thick] (0,-1,0) -- (0,5) node[anchor=south]{$u$};
    \draw[dotted] (3,5) -- (3,0) node[anchor=north]{$p_{FB}$};
    \draw[teal] (2,1.5) -- (0,1.5) node[anchor=east]{$\pi_s$};
    \draw[dotted] (0.84,5) -- (0.84,0) node[anchor=north]{$p_\times$} ;
    \draw[teal, dashed] (3,1.5) .. controls (3.5,1.5) and (4,2) .. (5,5) node[anchor=south]{$V_{FB}$};
    \draw (0,1.5) -- (5,1.5) node[anchor=west]{$\mathcal{D}_0$};
    \draw (0,4) -- (1.7,-1) node[anchor=west]{$\mathcal{D}_{N-1}$};
    \draw[dotted] (2.3,5) -- (2.3,0) node[anchor=north]{$p_I$} ;
    \draw[red] (2,1.5) .. controls (2.7,1.3) and (3.6,0.7) .. (4.8,4.2) node[anchor=west]{$u_A$};
    \draw[] (1,1.5) .. controls (1.4,1.2) and (3.6,0) .. (4.8,3.8) node[anchor=west]{$u_B$};
\end{tikzpicture}
    \caption{Value functions of multiple equilibria in an overcompetitive experimentation game. The dashed line $V_{FB}$ denotes the value function of the first-best solution, and $u_A$ and $u_B$ show the value functions of the agent in two different cutoff equilibrium of the noncooperative game. Note that $u_A$ is weakly above $u_B$ everywhere, and experimentation ends sooner (at a higher belief) in the equilibrium corresponding to $u_A$ than the equilibrium corresponding to $u_B$.}
    \label{fig:overcompetitive}
\end{figure}

Figure \ref{fig:overcompetitive} plots the value functions for the cutoff equilibria characterized by Proposition \ref{prop:overcompetitive}. As a side note, one might see from the figure that among the cutoff equilibria mentioned in Proposition \ref{prop:overcompetitive}, the agents strictly prefer equilibria where experimentation stops earlier (at higher beliefs). A quick examination of the value functions constructed in the proof of Proposition \ref{prop:overcompetitive} thus implies:
\begin{corollary}
    Suppose $\frac{\pi_s - \pi_l}{r} > R_l $. Let $u_{p_T}$ denote the value function of the symmetric equilibrium where all agents use cutoff strategies at $p_T$. If $p_T > p_T'$, then $u_{p_T}(p) \ge u_{p_T'}(p)$ $\forall p$, inequality holding strictly for $p > p_T'$. 
\end{corollary}

\section{Asymmetric and Nonmonotone Equilibria}
\label{app:other_equilibria}
In the paper, I focused on symmetric Markov perfect equilibria of the experimentation game in weakly monotonic strategies, particularly because the efficient solution requires the optimal policy to be symmetric in all agents and has a cutoff structure (and hence is weakly monotonic). In this appendix, I generalize some parts of Propositions \ref{prop:undercompetitive} and \ref{prop:overcompetitive} to discuss asymmetric and nonmonotone equilibria. 

\subsection{Undercompetitive Equilibria}
I first generalize Proposition \ref{prop:undercompetitive} to asymmetric equilibria. Note here that I require the finite piecewise Lipschitz assumption originally introduced in \cite{kr2010}; otherwise, the infinitely switching equilibrium of \cite{krc2005} results in experimentation until $p_{FB}$ (as \cite{krc2005} is a special case of my model). However, \cite{hkr2022} show that this infinitely switching equilibrium is a mathematical artifact of continuous time and never arises as the limit of discrete-time PBEs, so it is credible to exclude this equilibrium.
\begin{proposition}
    Suppose $\frac{\pi_s - \pi_l}{r} < R_l $. In \textbf{any} (symmetric or asymmetric, potentially nonmonotone) MPE, experimentation stops at $p_I$.
\end{proposition}
\begin{proof}
    In general, I cannot assume the value function is differentiable at $p_T$. Hence, I use the superdifferential/subdifferential generalizations of the derivative to show the result.

    I first argue that in any MPE, experimentation cannot stop above $p_I$. Let $u$ be the value function of a single agent. In order for experimentation to stop,
    \[ b_I(p, u, u') \le  c_I(p) \]
    \[ p\frac{\lambda}{r}(\pi_w - u(p) - (1-p)u'(p)) \le \pi_s - p \lambda R_w \]\
    \[ pu(p) + p(1-p)u'(p) \ge p \left( r R_w + \pi_w \right) - \frac{r \pi_s}{\lambda} \]
    
    Suppose experimentation stops at $p_T > p_I$. Take any sequence of $p \to p_T$ from below, and let $u'_-(p)$ be the smallest left subdifferential; that is, 
    \[ u'_-(p) = \min_{\{p_n\}, p_n < p_T} \lim_{p_n \to p_T} \frac{u(p_n) - u(p_T)}{p_n - p_T}  \]
    Noting that $u(p_T) = \pi_s$ when experimentation stops,
    \begin{align*}
        p_T(1-p_T)u'_-(p) &\ge p_T \left( r R_w + \pi_w - \pi_s \right) - \frac{r \pi_s}{\lambda} \\
        &= \frac{r}{\lambda}\left( \frac{p_T}{p_I}\pi_s - \pi_s \right)
    \end{align*}  
    
    Note that if $p_T > p_I$, then the RHS is positive, so $u'_-(p) > 0$. But this implies that for some $p$, $u(p) < \pi_s$, a contradiction, since playing $k_i = 0$ guarantees a payoff at least $\min(\pi_s, \pi_l + rR_l) \ge \pi_s$. 
    
    Now, I show experimentation in any MPE cannot stop below $p_I$. Suppose, for sake of contradiction, that in some permissible equilibrium, experimentation did stop at $p_T < p_I$. Let $i$ be an experimenter in that equilibrium who experiments until $p_T$, and suppose $i$'s value function is $u$. Since $i$ was experimenting, it must be the case that $b_I(p, u, u') \ge c_I(p)$ above $p_T$. Taking any sequence of $p \to p_T$ from above, let the maximum right superdifferential be 
    \[ u'_+(p) = \max_{\{p_n\}, p_n > p_T} \lim_{p_n \to p_T} \frac{u(p_n) - u(p_T)}{p_n - p_T} \]
    and noting that $u(p_T) = \pi_s$ when experimentation stops,
    \begin{align*}
        p_T(1-p_T)u'_+(p) &\le p_T \left( r R_w + \pi_w - \pi_s \right) - \frac{r \pi_s}{\lambda} \\
        &= \frac{r}{\lambda}\left( \frac{p_T}{p_I}\pi_s - \pi_s \right)
    \end{align*}  
    But since $p_T < p_I$, the RHS is negative, and so $u'_+$ must be negative. But again this implies that for some $p$, $u(p) < \pi_s$, a contradiction, since playing $k_i = 0$ guarantees a payoff at least $\min(\pi_s, \pi_l + rR_l) \ge \pi_s$. So in any MPE in weakly monotonic strategies, experimentation must stop at $p_I$.
\end{proof}

\subsection{Overcompetitive Equilibria}
If I retain the weak monotonicity assumption, I can also strengthen Proposition \ref{prop:overcompetitive} to asymmetric equilibria as well, since the bounds on the end of experimentation from Lemmas \ref{lem:stopping_upper_bound} and \ref{lem:stopping_lower_bound} required symmetry:
\begin{proposition}\label{prop:overcompetitive_2}
    Suppose $\frac{\pi_s - \pi_l}{r} > R_l $. In any symmetric or asymmetric Markov Perfect Equilibrium in weakly monotonic strategies, experimentation stops at some threshold $p_T \in [p_\times, p_I]$. 
\end{proposition}
\begin{proof}
    The argument is relatively simple; by Lemma \ref*{lem:stopping_upper_bound}, experimentation cannot stop above $p_I$. By Lemma \ref*{lem:stopping_lower_bound}, experimentation cannot stop at $p_T < p_\times$. Hence we are done. Note that neither of Lemmas \ref*{lem:stopping_upper_bound} and \ref*{lem:stopping_lower_bound} required strategies to be symmetric across agents.
\end{proof}

However, the weak monotonicity is a necessary condition for the result; if I drop the weak monotonicity condition, there is a special type of semi-efficient equilibria that can sometimes arise. In particular, ending experimentation at the belief $p_{FB}$ is sustained by a coordinated threat by all agents to exert effort on research below $p_{FB}$. This exists for some parameter values of the game; however, this condition is only partly efficient, since if the belief starts below $p_{FB}$ the agents experiment. Alternatively, this equilibrium is ``unstable''; if experimentation has stopped and the belief is at $p_{FB}$, even a small $\epsilon$ change to the belief downwards can induce experimentation (that is, given the arrival of very minor bad news about the project that all agents already stopped working on, the agents in this equilibrium have to start experimenting even though the bad news made them all more pessimistic about the project.) I pictorially depict this equilibrium for a particular choice of game parameters in Figure \ref{fig:semiefficient}.

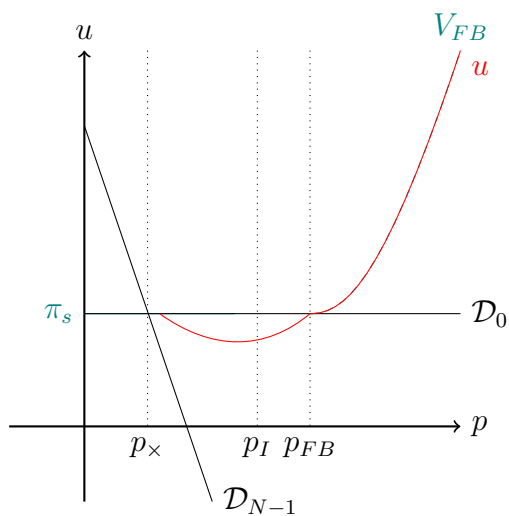
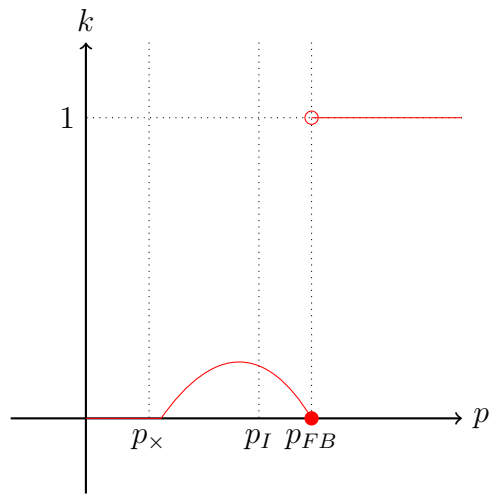
\begin{figure}
    \centering
    \begin{subfigure}[b]{0.45\textwidth}
         \centering
         \begin{tikzpicture}
    \draw[->, thick] (-1,0) -- (5,0) node[anchor=west]{$p$};
    \draw[->, thick] (0,-1,0) -- (0,5) node[anchor=south]{$u$};
    \draw[dotted] (3,5) -- (3,0) node[anchor=north]{$p_{FB}$};
    \draw[teal] (2,1.5) -- (0,1.5) node[anchor=east]{$\pi_s$};
    \draw[dotted] (0.84,5) -- (0.84,0) node[anchor=north]{$p_\times$} ;
    \draw[teal, dashed] (3,1.5) .. controls (3.5,1.5) and (4,2) .. (5,5) node[anchor=south]{$V_{FB}$};
    \draw (0,1.5) -- (5,1.5) node[anchor=west]{$\mathcal{D}_0$};
    \draw (0,4) -- (1.7,-1) node[anchor=west]{$\mathcal{D}_{N-1}$};
    \draw[dotted] (2.3,5) -- (2.3,0) node[anchor=north]{$p_I$} ;
    \draw[red] (3,1.5) .. controls (3.5,1.5) and (4,2) .. (5,5) node[anchor=north west]{$u$};
    \draw[red] (3,1.5) .. controls (2.4,1) and (1.7,1) .. (1,1.5);
\end{tikzpicture}
        \caption{Value functions. The dashed line $V_{FB}$ denotes the value function of the first-best solution, and $u$ denotes the value function of the agent in the semiefficient equilibrium. }
         \label{fig:semiefficient_value}
    \end{subfigure}
    \hfill
    \begin{subfigure}[b]{0.45\textwidth}
         \centering
         \begin{tikzpicture}
            \draw[->, thick] (-1,0) -- (5,0) node[anchor=west]{$p$};
            \draw[->, thick] (0,-1,0) -- (0,5) node[anchor=south]{$k$};
            \draw[dotted] (3,5) -- (3,0) node[anchor=north]{$p_{FB}$};
            \draw[dotted] (2.3,5) -- (2.3,0) node[anchor=north]{$p_I$} ;
            \draw[dotted] (0.84,5) -- (0.84,0) node[anchor=north]{$p_\times$} ;
            \draw[dotted] (5,4) -- (0,4) node[anchor=east]{1};
            \draw[red] (5,4) -- (3,4);
            \draw[red] (3,0) .. controls (2.4, 1) and (1.7, 1) .. (1,0);
            \draw[red] (1,0) -- (0,0);
            \filldraw[red] (3,0) circle (2.5pt);
            \draw[red] (3,4) circle (2.5pt);
        \end{tikzpicture}
        \caption{Strategies. Note that the unique symmetric equilibrium strategies are not weakly monotonic. Although the strategy exerts zero effort into research at $p_{FB}$, it exerts a positive amount of effort into research below $p_{FB}$. }
         \label{fig:semiefficient_strat}
    \end{subfigure}
    \caption{Semiefficient equilibria: value function and strategies. Note the ``efficiency'' (i.e. no experimentation at $p_{FB}$ is sustained by the threat to experiment at beliefs worse than $p_{FB}$. The semiefficient solution results in inefficiency if the belief ever ends up at any $p < p_{FB}$.}
    \label{fig:semiefficient}
\end{figure}

\end{document}